\documentclass[11pt]{article}

\usepackage[utf8]{inputenc}
\usepackage[english]{babel}
\usepackage[letterpaper,margin=1in]{geometry}
\usepackage{amsmath,amssymb,amsthm,mathtools}
\usepackage{enumitem}
\usepackage{booktabs,tabularx}
\usepackage{tikz}
\usetikzlibrary{decorations.pathreplacing,arrows.meta}
\definecolor{fpblue}{RGB}{46,92,150}
\definecolor{fpred}{RGB}{170,55,50}
\definecolor{fpgray}{gray}{0.55}
\definecolor{fplight}{gray}{0.90}
\tikzset{
  fp axis/.style={-{Stealth[length=2.2mm]}, line width=0.5pt, black},
  fp tick/.style={line width=0.4pt, black},
  fp main/.style={line width=1.1pt, black},
  fp accent/.style={line width=1.1pt, fpblue},
  fp accent dashed/.style={line width=1.1pt, fpblue, dash pattern=on 4pt off 2pt},
  fp guide/.style={line width=0.4pt, fpgray, dash pattern=on 2pt off 1.5pt},
  fp region/.style={fill=fplight, draw=none},
  fp accent region/.style={fill=fpblue, fill opacity=0.14, draw=none},
  fp label/.style={font=\small, inner sep=1.5pt},
  fp note/.style={font=\footnotesize, inner sep=1.5pt, align=center},
  fp panel/.style={font=\small},
  fp atom/.style={circle, fill=black, inner sep=0pt, minimum size=4.2pt},
  fp open/.style={circle, draw=black, fill=white, line width=0.5pt, inner sep=0pt, minimum size=4.2pt},
}

\definecolor{fplink}{gray}{0.15}
\usepackage[colorlinks,linkcolor=fplink,citecolor=fpblue,urlcolor=fpblue,hypertexnames=false]{hyperref}
\usepackage[capitalize]{cleveref}

\hypersetup{
 pdftitle={The Exact Approximation Ratio of the Optimal Fixed-Price Mechanism in Bilateral Trade},
 pdfauthor={Tao Jiang, Minbo Gao, Shaowei Cai},
 pdfsubject={Exact worst-case welfare ratio of deterministic fixed-price bilateral trade},
 pdfkeywords={bilateral trade, fixed-price mechanisms, welfare approximation, convex variational problems}
}

\allowdisplaybreaks

\newtheorem{theorem}{Theorem}[section]
\newtheorem{lemma}[theorem]{Lemma}
\newtheorem{corollary}[theorem]{Corollary}
\newtheorem{proposition}[theorem]{Proposition}
\newtheorem{claim}[theorem]{Claim}
\theoremstyle{definition}

\theoremstyle{remark}
\newtheorem{remark}[theorem]{Remark}

\newcommand{\E}{\mathbb{E}}
\newcommand{\Pp}{\mathbb{P}}

\newcommand{\one}{\mathbf{1}}
\newcommand{\OPT}{\mathrm{OPT}}
\newcommand{\FP}{\mathrm{FP}}
\newcommand{\SW}{\mathrm{SW}}
\newcommand{\dd}{\,\mathrm{d}}
\newcommand{\K}{\mathcal{K}}
\newcommand{\Acal}{\mathcal{A}}
\newcommand{\Qcal}{\mathcal{Q}}
\newcommand{\Ucal}{\mathcal{U}}
\newcommand{\Ycal}{\mathcal{Y}}
\newcommand{\afp}{\alpha_{\mathrm{FP}}}

\begin{document}

\title{The Exact Approximation Ratio of the Optimal Fixed-Price Mechanism in Bilateral Trade}
\author{%
Tao Jiang\thanks{Tao Jiang and Minbo Gao contributed equally to this work.}
\qquad Minbo Gao\footnotemark[1]
\qquad Shaowei Cai\\[0.6em]
\small Key Laboratory of System Software (Chinese Academy of Sciences)\\
\small State Key Laboratory of Computer Science\\
\small Institute of Software, Chinese Academy of Sciences\\
\small School of Computer Science and Technology, University of Chinese Academy of Sciences\\
\small Beijing, China\\[0.3em]
\small\texttt{\char123 jiangt,gaomb,caisw\char125 @ios.ac.cn}%
}
\date{}

\maketitle

\begin{abstract}
In bilateral trade, a seller owns an item that a buyer may value more, and their values are drawn independently from known distributions. A fixed-price mechanism posts a single price and trades exactly when both agents accept it. How much of the first-best expected welfare can the optimal fixed price guarantee? A sequence of analytic and computer-assisted results narrowed the worst-case ratio to the interval $[0.7292,0.73805]$, but its exact value remained open.

We determine this ratio exactly: it equals $0.7380243357\ldots$, given in closed form through the unique root of an explicit one-dimensional equation. The guarantee holds for every pair of independent nonnegative distributions with finite means, including atomic and unbounded ones, and an explicit family of instances shows that no larger constant is possible.

The proof reduces the problem, through put and call transforms of the two distributions and a worst-case choice of seller, to a one-dimensional control problem for the buyer. A logarithmic change of variables makes this problem strictly convex, so its first-order conditions certify a global optimum. The optimizer switches exactly once from an interior arc to the boundary and can be integrated in closed form. Realizing it yields a bounded seller and a bounded buyer together with a vanishing buyer atom at a diverging value. The worst case is approached by this family but attained by no instance.
\end{abstract}

\section{Introduction}\label{sec:intro}

Bilateral trade is the most basic market in mechanism design. A seller owns an indivisible item, a buyer may acquire it, and each agent privately knows a value: $S$ for the seller to keep the item and $B$ for the buyer to receive it. Social welfare is maximized by trading exactly when $B\ge S$. The theorem of Myerson and Satterthwaite~\cite{MyersonSatterthwaite1983} shows that no mechanism that is incentive compatible, individually rational, and budget balanced achieves this first-best allocation in general. Welfare must be lost, and the question is how much.

Fixed-price mechanisms are the natural benchmark. The designer posts a price $p$; the seller sells if $S\le p$, the buyer buys if $B\ge p$, and trade occurs if both agree. Such a mechanism is dominant-strategy incentive compatible, ex post individually rational, and strongly budget balanced, and Hagerty and Rogerson~\cite{HagertyRogerson1987} showed that, essentially, every deterministic mechanism with these properties is a posted price. Fixed prices are therefore not one simple heuristic among many; they are the mechanisms that remain once one insists on robust incentives and exact budget balance.

\subsection{The fixed-price welfare ratio}

We work in the full-prior model with independent values. Let $F_S$ and $F_B$ be Borel probability distributions on $[0,\infty)$ with finite first moments, and let $S\sim F_S$ and $B\sim F_B$ be independent. Knowing both distributions, the designer selects a price $p\ge0$ before the values are realized, and trade occurs exactly when $S\le p\le B$.
The first-best welfare, the welfare of price $p$, and the optimal fixed-price welfare are
\begin{align*}
\OPT(F_S,F_B)&:=\E[\max\{S,B\}],\\
\SW(p;F_S,F_B)
&:=\E\!\left[B\one\{S\le p\le B\}+S\one\{\text{not }(S\le p\le B)\}\right],\\
\FP(F_S,F_B)&:=\max_{p\ge0}\SW(p;F_S,F_B).
\end{align*}
The maximum is attained for every admissible pair (\cref{lem:price-attainment}). The object of study is the worst-case ratio
\begin{equation}\label{eq:alpha-def}
\afp
:=\inf_{F_S,F_B:\ \OPT>0}\frac{\FP(F_S,F_B)}{\OPT(F_S,F_B)},
\end{equation}
where the infimum ranges over all admissible pairs with $\OPT(F_S,F_B)>0$.

\subsection{Prior bounds and the open question}\label{sec:related-work}

A sequence of papers has bounded $\afp$ from both sides; \cref{tab:related-work} summarizes them. Simple choices of price already give constant guarantees: posting a price drawn from the seller's distribution secures half of the optimal welfare~\cite{DuettingFLLR2021}, and Blumrosen and Dobzinski showed that a suitably randomized seller quantile secures a $1-1/e$ fraction~\cite{BlumrosenDobzinski2021}. Since a randomized price is an average of deterministic ones, the optimal fixed price is at least as good. Kang, Pernice, and Vondr\'ak~\cite{KangPerniceVondrak2022} proved that $1-1/e$ is not the answer for the optimal price with full knowledge of both distributions, improving the guarantee to $1-1/e+10^{-4}$. They also determined the exact ratio $(2+\sqrt2)/4$ for identically distributed values and left the asymmetric case open.

The asymmetric problem was then attacked by optimization. Liu, Ren, and Wang~\cite{LiuRenWang2023} identified, for each buyer distribution, the worst-case seller distribution. They reduced the question of whether a given $\beta$ is a valid guarantee to the maximization of a functional over buyer survival functions, and they derived a first-order condition for its maximizer. They noted, however, that an explicit solution of this condition appeared hard to obtain. Instead, they certified the bounds $0.71\le\afp\le0.7381$ through discretization and dynamic programming. Independently, Cai and Wu~\cite{CaiWu2023} formulated the worst case as an infinite-dimensional factor-revealing program. They obtained $0.72\le\afp\le0.7381$ by solving convergent finite discretizations numerically. Most recently, Giambartolomei and de Keijzer~\cite{GiambartolomeiDeKeijzer2026} sharpened both sides with further mathematical programming and explicit hard instances, which left the interval $0.7292\le\afp\le0.73805$.

\begin{table}[t]
\centering
\small
\caption{Bounds on the worst-case welfare ratio $\afp$ of the optimal fixed price for independent, non-identical values. A dash indicates that no upper bound for this setting is given.}\label{tab:related-work}
\begin{tabularx}{\textwidth}{@{}l l l >{\raggedright\arraybackslash}X@{}}
\toprule
Work & Lower bound & Upper bound & Method \\
\midrule
D\"utting et al.~\cite{DuettingFLLR2021} & $1/2$ & -- & seller-sample price \\
Blumrosen--Dobzinski~\cite{BlumrosenDobzinski2021} & $1-1/e$ & -- & random seller quantile \\
Kang--Pernice--Vondr\'ak~\cite{KangPerniceVondrak2022} & $1-1/e+10^{-4}$ & -- & analytic quantile argument \\
Liu--Ren--Wang~\cite{LiuRenWang2023} & $0.71$ & $0.7381$ & worst-case seller; dynamic programming \\
Cai--Wu~\cite{CaiWu2023} & $0.72$ & $0.7381$ & factor-revealing program; discretization \\
Giambartolomei--de Keijzer~\cite{GiambartolomeiDeKeijzer2026} & $0.7292$ & $0.73805$ & mathematical programming \\
\midrule
This paper & \multicolumn{2}{l}{$0.7380243357\ldots$ (exact)} & analytic solution of a convex control problem \\
\bottomrule
\end{tabularx}
\end{table}

The computational bounds are rigorous; each comes with a certificate. They still leave a gap, however, and a finer computation can shrink it but cannot close it. Closing it requires two things that discretization does not supply: an exact description of the infinite-dimensional worst case, and a proof that this description is globally optimal rather than merely stationary. Determining the exact value of $\afp$ is the question we resolve.

\subsection{Main result}

For $1<r<3$, define
\begin{equation}\label{eq:theta-intro}
\Theta(r):=\int_0^{(r-1)/r}\frac{\dd x}{x^2-x+\frac{1}{r+1}}.
\end{equation}
The denominator is positive on the integration interval exactly when $r<3$, so $\Theta$ is finite on $(1,3)$; a closed form in terms of arctangents is given in \cref{eq:Theta-closed}.

\begin{theorem}[Exact worst-case fixed-price welfare ratio]\label{thm:main}
There is a unique $r_\star\in(1,3)$ satisfying
\begin{equation}\label{eq:root-intro}
e^{-\Theta(r_\star)/2}
=(r_\star-1)\bigl(\Theta(r_\star)-r_\star+1\bigr),
\end{equation}
and
\begin{equation}\label{eq:alpha-intro}
\afp=\frac{r_\star+1}{\Theta(r_\star)+2}
=0.7380243357\ldots.
\end{equation}
More precisely:
\begin{enumerate}[label=(\alph*),leftmargin=2em]
\item \emph{(Universal guarantee.)} For every pair of independent, nonnegative distributions $F_S,F_B$ with finite first moments, $\FP(F_S,F_B)\ge\afp\cdot\OPT(F_S,F_B)$.
\item \emph{(Tightness.)} There are an explicit bounded seller distribution $F_{S_\star}$ and an explicit family of finite-mean buyer distributions $(F_{B_\varepsilon})_{\varepsilon>0}$ such that $\FP(F_{S_\star},F_{B_\varepsilon})/\OPT(F_{S_\star},F_{B_\varepsilon})\to\afp$ as $\varepsilon\downarrow0$.
\item \emph{(Non-attainment.)} No admissible pair attains the infimum in \cref{eq:alpha-def}, although every individual pair admits an optimal price.
\end{enumerate}
\end{theorem}

Part~(a) says that no instance forces the optimal fixed price below the constant, and part~(b) says that the constant cannot be raised; together they identify $\afp$ exactly. The value lies just below the previous upper bound $0.73805$ of~\cite{GiambartolomeiDeKeijzer2026}, so the hard instances found by computation were already nearly optimal; the substance of the theorem is that the matching lower bound holds with no loss. Interval arithmetic gives $0.73802433573449<\afp<0.73802433573450$ (\cref{sec:numerics}). These decimals only evaluate the closed form; they play no role in the proof.

Beyond the constant itself, the proof describes the worst case completely. The extremal control problem has a unique optimizer in closed form, and its global optimality is certified by convexity rather than inferred from stationarity. The hard instances consist of a bounded seller with an atom at zero, a bounded continuous buyer body, and a vanishing buyer atom at a value that tends to infinity. The last component explains non-attainment: the atom carries a fixed amount of first moment while its probability tends to zero, so the worst case is a limit of instances in which first moment escapes to infinity, and no single instance realizes it. No continuity, density, or bounded-support assumption is made anywhere.

\subsection{Approach}\label{sec:proof-overview}

The quantity $\afp$ is the value of a nonconvex optimization over pairs of probability measures, and three difficulties separate the earlier bounds from an exact answer. Once the seller is eliminated, the problem over buyer distributions is still nonconvex in its natural coordinates, so first-order conditions such as those of~\cite{LiuRenWang2023} are necessary for optimality but cannot exclude better solutions elsewhere. Even a certified optimum is useful only if the optimality system can be solved, which requires knowing where the optimizer meets the constraint that probabilities are at most one. Finally, natural relaxations drop feasibility constraints such as monotonicity of a survival function, so the relaxed optimizer must be shown to come from genuine distributions, and the supremum must be shown to be approached although it is not attained.

To describe how the proof meets these difficulties, write $m:=\E[S]$, let $I:=\E[(B-S)_+]$ be the first-best gain from trade, and let $M$ be the gain captured by the optimal fixed price. Then $\OPT=m+I$ and $\FP=m+M$, so a guarantee $(1+h)^{-1}$, with $h>0$, is equivalent to the linear inequality
\begin{equation}\label{eq:overview-linear}
I-hm\le(1+h)M.
\end{equation}
We encode the seller by its put transform $A(p)=\E[(p-S)_+]$ and the buyer by its call transform $C(p)=\E[(B-p)_+]$. At almost every price the captured gain is the Wronskian $A'C-AC'$, an atom-aware version holds at every price, and $I=\E[A(B)]$ (\cref{sec:transforms}). After normalizing $M=1$, the Wronskian bound reads $(A/C)'\le C^{-2}$. Saturating it pointwise on the range where $C\ge1$ defines, for each buyer, the seller transform
\[
\widehat A(p)=C(p)\int_0^p\frac{\dd t}{C(t)^2},
\]
and replacing the seller by this one does not decrease the left-hand side of \cref{eq:overview-linear}. This is the worst-case seller of Liu, Ren, and Wang~\cite[Theorems~3.1 and~3.5]{LiuRenWang2023}: with their $G_B=C$ and $H_B=Q_B$, where $Q_B(p):=\Pp(B>p)$ is the buyer's survival function, its distribution function $\widehat A'=1/C-Q_B\int_0^p C^{-2}$ is their formula after the constant captured gain is normalized to one. Our derivation is at the level of transforms and applies directly to arbitrary distributions, including atoms (\cref{sec:saturation}).

The buyer is then described in inverse-call coordinates: the call value $c=C(p)$ serves as the independent variable and the survival probability $q(c)=\Pp(B>p(c))$ as the control. The left-hand side of \cref{eq:overview-linear} becomes $1+\K_h(b,q)$ for an explicit action $\K_h$, whose supremum $V(h)$ over endpoints $b>1$ and measurable controls $0<q\le1$ satisfies $I-hm\le1+V(h)$ (\cref{sec:control-coords}). The only relaxation is that $q$ need not be monotone, as in the first-order analysis of~\cite[Section~3]{LiuRenWang2023}. A guarantee $(1+h)^{-1}$ therefore holds whenever $V(h)\le h$; \cref{prop:decision-equivalence} shows that this condition is also necessary, which is the counterpart in our coordinates of the threshold characterization in~\cite[Theorem~3.3]{LiuRenWang2023}.

The action is not convex in $q$, but the substitution $q=e^u$ turns it into
\[
h\int (1-c^{-2})e^{-u(c)}\dd c
+\iint_{\sigma\le c}\frac{e^{u(\sigma)-u(c)}}{c^2}\dd\sigma\dd c,
\]
a sum of exponentials of linear functionals of $u$ that is strictly convex because of the first term. Its KKT conditions are therefore sufficient for global optimality as well as necessary (\cref{sec:convexity}), which removes the first difficulty. They express $q$ through a switching function that can cross zero only upward, so the optimizer consists of a single interior arc followed by the boundary arc $q=1$ (\cref{sec:one-switch}, \cref{fig:one-switch}); this removes the second. On the interior arc the optimality system has a first integral, and the trajectory is determined by the ratio $r=b/s$ of the endpoint to the switching point. Explicit integration yields monotone functions $H_-$ and $H_+$ with $h=H_-(r)$ and $V(h)=H_+(r)$, and their unique crossing is the fixed point $V(h_\star)=h_\star$, which gives \cref{eq:root-intro} and $\afp=(1+h_\star)^{-1}$ (\cref{sec:integration,sec:fixed-point}).

For the third difficulty, the optimal control turns out to be monotone, so it is the survival function of a bounded buyer body, and the saturated seller is a genuine bounded seller (\cref{sec:extremal}). This pair realizes the control value exactly but has call transform $C-1$ rather than $C$. A buyer atom of mass $\varepsilon$ at value $1/\varepsilon$ restores the missing unit: as $\varepsilon\to0$ it vanishes in distribution, yet it contributes one unit of first moment and of first-best gain, and the welfare ratios converge to $\afp$ (\cref{fig:extremal}). Liu, Ren, and Wang already use a very large buyer value with very small probability to realize a boundary case of their formulation~\cite{LiuRenWang2023}. Here the device is part of the exact extremal structure, and it is also what rules out attainment. An instance attaining the bound would need a buyer survival function that vanishes at the price where the call transform equals one, which is impossible because the call transform is still positive there.

\subsection{Further related work and organization}

\Cref{thm:main} concerns deterministic fixed prices in the full-prior model with independent values, and neighboring questions have different answers. Kang, Pernice, and Vondr\'ak~\cite{KangPerniceVondrak2022} and Cai and Wu~\cite{CaiWu2023} also study designers who know only one of the two distributions, only samples, or only means. Approximating the gains from trade rather than the welfare is a separate problem with its own literature, on which our results are silent. Mechanisms outside the fixed-price class can do better: Dobzinski and Shaulker~\cite{DobzinskiShaulker2026} show that buyer-offering mechanisms with a reserve guarantee a $0.746$ fraction of the first-best welfare, so the constant in \cref{thm:main} is specific to fixed prices.

The paper follows the order of the argument above. \Cref{sec:transforms} records the welfare identities and the transforms, and \cref{sec:control-reduction} reduces an arbitrary instance to the control problem defining $V(h)$. \Cref{sec:variational} solves that problem globally, and \cref{sec:fixed-point} derives the root equation and the universal guarantee. \Cref{sec:extremal} constructs the matching family, proves non-attainment and the threshold equivalence, and completes the proof of \cref{thm:main}. All proofs are given in full in the main text.
\section{Preliminaries and transform geometry}\label{sec:transforms}

This section fixes conventions, decomposes both welfare benchmarks into a common seller baseline plus a gain from trade, and expresses the fixed-price gain through two convex transforms of the distributions. Everything here is elementary, but the atom-aware formulas are what later allow the theorem to cover arbitrary distributions.

\subsection{Conventions and welfare decomposition}

Throughout, $S$ and $B$ are independent, nonnegative random variables with finite first moments. For every nonnegative random variable $X$, we use the right-continuous CDF and survival conventions $F_X(p):=\Pp(X\le p)$ and $Q_X(p):=\Pp(X>p)$. As usual, $F_X$ denotes both the law of $X$ and its distribution function.
When an endpoint atom is relevant, $\Pp(X\ge p)=Q_X(p)+\Pp(X=p)$ is written explicitly. Equalities involving ordinary derivatives are asserted almost everywhere; one-sided derivatives are used for pointwise endpoint statements.

Let
\begin{equation}\label{eq:m-I-g-M}
m:=\E[S],\qquad I:=\E[(B-S)_+],\qquad
g(p):=\E[(B-S)\one\{S\le p\le B\}],\qquad M:=\max_{p\ge0}g(p).
\end{equation}
Thus $I$ is the first-best expected gain from trade and $M$ is the gain captured by the optimal fixed price.

\begin{lemma}[Welfare decomposition]\label{lem:welfare-decomp}
For every admissible pair,
\begin{equation}\label{eq:opt-fp-decomp}
\OPT=m+I,\qquad \FP=m+M.
\end{equation}
In particular, $0\le M\le I<\infty$.
\end{lemma}

\begin{proof}
The seller keeps the item in the no-trade allocation, so its welfare is $S$. Relative to that baseline, an efficient trade adds $(B-S)_+$, while a fixed price $p$ adds $(B-S)\one\{S\le p\le B\}$. Taking expectations gives \cref{eq:opt-fp-decomp}. The inequality $M\le I$ follows because the fixed-price trade event is contained in $\{B\ge S\}$.
\end{proof}

\Cref{fig:trade-regions} compares the two trade events. A fixed price trades only on the product set $\{S\le p\}\times\{B\ge p\}$, which is contained in the efficient event $\{B\ge S\}$. The gain $M$ therefore misses the efficient trades with $S>p$ and those with $B<p$; the ratio $\afp$ measures how much of $I$ these losses can cost in the worst case.

\begin{figure}[!htb]
\centering
% Figure: first-best vs fixed-price trade regions in the (S,B) value plane.
% Schematic: axes are values, not probabilities; areas carry no welfare meaning.
\begin{tikzpicture}[x=4.1cm,y=4.1cm]
  % ---------- panel (a): first best ----------
  \begin{scope}
    \fill[fp region] (0,0) -- (1,1) -- (0,1) -- cycle;
    \draw[fp main] (0,0) -- (1,1);
    \draw[fp axis] (0,0) -- (1.1,0) node[right, fp label] {$S$};
    \draw[fp axis] (0,0) -- (0,1.1) node[above, fp label] {$B$};
    \node[fp label, rotate=45, above=2pt] at (0.80,0.80) {$B=S$};
    \node[fp note] at (0.30,0.78) {trade\\[-1pt]$B\ge S$};
    \node[fp note] at (0.74,0.26) {no trade};
    \node[fp panel] at (0.55,-0.16) {(a) first-best};
  \end{scope}
  % ---------- panel (b): fixed price p ----------
  \begin{scope}[shift={(1.55,0)}]
    \def\p{0.42}
    \fill[fp region] (0,0) -- (1,1) -- (0,1) -- cycle;
    \fill[fp accent region] (0,\p) rectangle (\p,1);
    \draw[fpblue, line width=0.9pt] (0,\p) -- (\p,\p) -- (\p,1);
    \draw[fp main] (0,0) -- (1,1);
    \draw[fp guide] (\p,0) -- (\p,\p);
    \draw[fp axis] (0,0) -- (1.1,0) node[right, fp label] {$S$};
    \draw[fp axis] (0,0) -- (0,1.1) node[above, fp label] {$B$};
    \draw[fp tick] (\p,0) -- (\p,-0.025) node[below, fp label] {$p$};
    \draw[fp tick] (0,\p) -- (-0.025,\p) node[left, fp label] {$p$};
    \node[fp note, text=fpblue] at (0.21,0.74) {trade\\[-1pt]$S\le p\le B$};
    % efficient trades the price misses (still in B >= S)
    \node[fp note, fill=fplight] at (0.60,0.86) {missed:\\[-1pt]$S>p$};
    \node[fp note, fill=fplight] at (0.15,0.25) {missed:\\[-1pt]$B<p$};
    \node[fp panel] at (0.55,-0.16) {(b) fixed price $p$};
  \end{scope}
\end{tikzpicture}
\caption{Trade events in the plane of values $(S,B)$. (a)~The first-best allocation trades on the half-plane $B\ge S$ (shaded). (b)~A fixed price $p$ trades exactly on $\{S\le p\le B\}$ (blue quadrant). This is a subset of $\{B\ge S\}$, so efficient trades with $S>p$ or $B<p$ are lost. The axes carry values, not probabilities; areas in the diagram have no welfare meaning.}\label{fig:trade-regions}
\end{figure}

\begin{lemma}[Attainment of an optimal price]\label{lem:price-attainment}
For every admissible pair, $g$ is upper semicontinuous on $[0,\infty)$ and satisfies $g(p)\to0$ as $p\to\infty$. Consequently, $g$ attains its maximum and the supremum in the definition of the optimal fixed-price welfare is a maximum.
\end{lemma}

\begin{proof}
For fixed values $(v_S,v_B)$, the function $p\mapsto (v_B-v_S)\one\{v_S\le p\le v_B\}$
is identically zero when $v_B<v_S$, and otherwise is a nonnegative multiple of the indicator of the closed interval $[v_S,v_B]$. It is therefore upper semicontinuous. It is dominated by the integrable random variable $(B-S)_+$. Hence, for every sequence $p_n\to p$, reverse Fatou gives
\[
\limsup_{n\to\infty}g(p_n)
\le \E\!\left[\limsup_{n\to\infty}(B-S)\one\{S\le p_n\le B\}\right]
\le g(p).
\]
Thus $g$ is upper semicontinuous. For the behavior at infinity,
\[
0\le g(p)\le \E\bigl[B\one\{B\ge p\}\bigr]\longrightarrow0
\qquad(p\to\infty),
\]
by integrability of $B$. If $M=0$, every price attaining zero is optimal. If $M>0$, choose $R$ so large that $g(p)<M/2$ for $p\ge R$; upper semicontinuity then gives a maximizer on the compact interval $[0,R]$.
\end{proof}

\subsection{Put and call transforms}

The fixed-price gain depends on the two distributions only through two integrated tails. Define the seller put transform and buyer call transform by $A(p):=\E[(p-S)_+]$ and $C(p):=\E[(B-p)_+]$ for $p\ge0$.
Both are finite and convex. The function $A$ is nondecreasing and $1$-Lipschitz, while $C$ is nonincreasing and $1$-Lipschitz. Their one-sided derivatives satisfy
\begin{align}
A'_+(p)&=\Pp(S\le p),&A'_-(p)&=\Pp(S<p),\label{eq:A-derivatives}\\
-C'_+(p)&=\Pp(B>p),&-C'_-(p)&=\Pp(B\ge p).\label{eq:C-derivatives}
\end{align}

Conversely, a function with these properties is the put or call transform of a distribution. The following two characterizations will be used to recognize the seller and buyer constructed in \cref{sec:saturation,sec:extremal} as genuine random variables.

\begin{lemma}[Put-transform characterization]\label{lem:put-characterization}
Let $A:[0,\infty)\to[0,\infty)$ be convex with $A(0)=0$. Suppose its right derivative $F=A'_+$ is nondecreasing, right-continuous, takes values in $[0,1]$, and satisfies $F(p)\to1$ as $p\to\infty$. If $\int_0^\infty(1-F(p))\dd p<\infty$, then there is a nonnegative finite-mean random variable $S$ with CDF $F$ such that $A(p)=\E[(p-S)_+]$ for all $p\ge0$.
Moreover,
\[
\E[S]=\int_0^\infty(1-F(p))\dd p=\lim_{p\to\infty}(p-A(p)).
\]
\end{lemma}

\begin{proof}
The assumptions make $F$ a CDF on $[0,\infty)$. Convexity and $A(0)=0$ give $A(p)=\int_0^pF(t)\dd t$. For a random variable with CDF $F$, Tonelli gives
\[
\E[(p-S)_+]=\int_0^p\Pp(S\le t)\dd t=\int_0^pF(t)\dd t.
\]
The mean identities follow from the tail integral formula and $p-A(p)=\int_0^p(1-F(t))\dd t$.
\end{proof}

\begin{lemma}[Call-transform characterization]\label{lem:call-characterization}
Let $C:[0,\infty)\to[0,\infty)$ be convex, nonincreasing, and satisfy $C(p)\to0$. Suppose $-C'_+(p)\in[0,1]$. Then $Q(p):=-C'_+(p)$ is the survival function of a nonnegative integrable random variable $B$, and $C(p)=\E[(B-p)_+]$.
\end{lemma}

\begin{proof}
Convexity makes $C'_+$ nondecreasing, hence $Q$ nonincreasing. Since $C$ is finite and tends to zero, $Q(p)\to0$. The identity $C(p)=\int_p^\infty Q(t)\dd t$ follows by absolute continuity and the boundary condition at infinity. This is the integrated-survival representation of a call transform.
\end{proof}

\begin{lemma}[Atom-aware fixed-price identity]\label{lem:atom-aware}
For every $p\ge0$,
\begin{equation}\label{eq:g-atom-aware}
g(p)=A(p)\Pp(B\ge p)+C(p)\Pp(S\le p).
\end{equation}
At every point where both $A$ and $C$ are differentiable,
\begin{equation}\label{eq:wronskian}
g(p)=A'(p)C(p)-A(p)C'(p).
\end{equation}
Moreover,
\begin{equation}\label{eq:I-identities}
I=\E[A(B)]=\int_0^\infty\Pp(S\le p)\Pp(B>p)\dd p.
\end{equation}
\end{lemma}

\begin{proof}
On the event $\{S\le p\le B\}$, $B-S=(p-S)+(B-p)$. Independence therefore gives
\begin{align*}
g(p)
&=\E[(p-S)\one\{S\le p\}]\Pp(B\ge p)
+\Pp(S\le p)\E[(B-p)\one\{B\ge p\}]\\
&=A(p)\Pp(B\ge p)+C(p)\Pp(S\le p),
\end{align*}
which proves \cref{eq:g-atom-aware}. At a common differentiability point there are no atoms at $p$, and \cref{eq:A-derivatives,eq:C-derivatives} gives \cref{eq:wronskian}.

For the first equality in \cref{eq:I-identities}, condition on $B$ and use independence. For the second, use
\[
(B-S)_+=\int_0^\infty\one\{S\le p<B\}\dd p
\]
and apply Tonelli's theorem.
\end{proof}

The ratio $\FP/\OPT$ is invariant under a common positive rescaling of $S$ and $B$. The following reformulation turns the ratio guarantee into a linear inequality and will be used throughout.

\begin{lemma}[Linear target inequality]\label{lem:linear-target}
Fix $h\ge0$ and set $\alpha=(1+h)^{-1}$. The universal bound $\FP\ge\alpha\OPT$ is equivalent to
\begin{equation}\label{eq:linear-target}
I-hm\le(1+h)M.
\end{equation}
\end{lemma}

\begin{proof}
By \cref{lem:welfare-decomp}, \cref{eq:linear-target} is equivalent to $m+I\le(1+h)(m+M)$, which is exactly $(m+M)/(m+I)\ge(1+h)^{-1}$ whenever $\OPT=m+I>0$.
\end{proof}

\section{Reduction to a variational problem}\label{sec:control-reduction}

By \cref{lem:linear-target}, the guarantee $(1+h)^{-1}$ amounts to bounding $I-hm$ by $(1+h)M$ for every instance. This section shows that $I-hm\le1+V(h)$ after normalizing $M=1$, where $V(h)$ is the value of an explicit one-dimensional control problem. The reduction has two steps. First, for a fixed buyer, the seller is replaced by a worst-case seller that makes the fixed-price gain constant (\cref{sec:saturation}). Second, the buyer is described through inverse-call coordinates, in which the objective becomes an action functional of a single control (\cref{sec:control-coords}). Only the second step involves a relaxation, and we identify it explicitly.

\subsection{Wronskian saturation: the worst-case seller}\label{sec:saturation}

Assume first that $M>0$ and rescale values so that $M=1$.
If $M=0$, then $I=0$: otherwise the integral in \cref{eq:I-identities} is positive, so for some $p$ both $\Pp(S\le p)>0$ and $C(p)>0$, and \cref{eq:g-atom-aware} gives $g(p)>0$. Thus an instance with $M=0$ has welfare ratio one.

At almost every $p$, \cref{eq:wronskian} and the normalization $M=1$ imply
\begin{equation}\label{eq:wronskian-bound}
A'(p)C(p)-A(p)C'(p)\le1.
\end{equation}

If $C(0)=\E[B]\le1$, then $I=\E[A(B)]\le\E[B]=C(0)\le1$, so $I-hm\le1<1+h$ for every $h>0$. It remains to consider $C(0)>1$. Because $C$ is continuous, strictly decreasing on the set where $C>0$, and tends to zero, there is a unique $p_1>0$ such that $C(p_1)=1$.

Define, for $0\le p\le p_1$,
\begin{equation}\label{eq:R-Ahat}
R(p):=\int_0^p\frac{\dd t}{C(t)^2},\qquad \widehat A(p):=C(p)R(p).
\end{equation}

\begin{lemma}[Saturated seller transform]\label{lem:saturation}
The function $\widehat A$ has the following properties on $[0,p_1]$.
\begin{enumerate}[label=(\roman*),leftmargin=2em]
\item $A(p)\le\widehat A(p)$ for all $p\in[0,p_1]$.
\item $\widehat A$ is convex, $\widehat A(0)=0$, and $0\le\widehat A'(p)\le1$ for almost every $p\in[0,p_1]$.
\item At almost every $p\in[0,p_1]$,
\begin{equation}\label{eq:saturated-wronskian}
\widehat A'(p)C(p)-\widehat A(p)C'(p)=1.
\end{equation}
\end{enumerate}
Consequently, the slope-one extension
\begin{equation}\label{eq:Atilde}
\widetilde A(p):=
\begin{cases}
\widehat A(p),&0\le p\le p_1,\\
\widehat A(p_1)+p-p_1,&p\ge p_1
\end{cases}
\end{equation}
is the put transform of a nonnegative random variable $\widetilde S$ supported on $[0,p_1]$.
\end{lemma}

\begin{proof}
Because $C\ge1$ on $[0,p_1]$, the quotient $A/C$ is absolutely continuous. By \cref{eq:wronskian-bound},
\[
\left(\frac{A}{C}\right)'=\frac{A'C-AC'}{C^2}\le\frac1{C^2}=R'
\quad\text{almost everywhere}.
\]
Both $A/C$ and $R$ vanish at zero, hence $A/C\le R$ and $A\le\widehat A$.

At differentiability points, $-C'(p)=Q_B(p)=\Pp(B>p)$, and therefore
\begin{equation}\label{eq:Ahat-prime}
\widehat A'(p)=\frac1{C(p)}-Q_B(p)R(p).
\end{equation}
Since $Q_B$ is nonincreasing, for almost every $p$ with $Q_B(p)>0$,
\[
Q_B(p)R(p)\le\int_0^p\frac{Q_B(t)}{C(t)^2}\dd t
=\frac1{C(p)}-\frac1{C(0)}.
\]
The same inequality is trivial if $Q_B(p)=0$. Therefore $\widehat A'(p)\ge1/C(0)>0$. Also \cref{eq:Ahat-prime} gives $\widehat A'(p)\le1/C(p)\le1$.

To see convexity in the presence of atoms, write $\widehat A'=C'R+1/C$. The distributional product rule gives
\begin{align}
\dd\widehat A'
&=R\,\dd C'+C'\,\dd R+\dd(1/C)\notag\\
&=R\,\dd C'+\frac{C'}{C^2}\dd p-\frac{C'}{C^2}\dd p
=R\,\dd C'.
\label{eq:Ahat-second}
\end{align}
Here $R$ and $1/C$ are absolutely continuous, while $\dd C'$ is the nonnegative second-derivative measure of the convex function $C$. Hence $\widehat A$ is convex. Differentiating \cref{eq:R-Ahat} directly gives \cref{eq:saturated-wronskian}.

The left derivative of $\widehat A$ at $p_1$ is at most one, so the extension \cref{eq:Atilde} remains convex. Its right derivative is a nondecreasing function taking values in $[0,1]$ and equal to one for $p\ge p_1$. Defining this derivative as a cumulative distribution function produces a probability distribution on $[0,p_1]$, and
\[
\widetilde A(p)=\int_0^pF_{\widetilde S}(t)\dd t=\E[(p-\widetilde S)_+].
\]
This is the content of \cref{lem:put-characterization}.
\end{proof}

\begin{remark}[Relation to the worst-case seller of Liu, Ren, and Wang]\label{rem:lrw-seller}
For a fixed buyer, Liu, Ren, and Wang~\cite[Theorems~3.1 and~3.5]{LiuRenWang2023} identify a worst-case seller whose distribution makes the gain captured by every price below a threshold equal to a constant. In their notation $H_B=Q_B$ and $G_B=C$, and after the constant gain is normalized to one their seller distribution function is exactly \cref{eq:Ahat-prime}. Their threshold $s_2$ corresponds to our $p_1$, defined by $C(p_1)=1$. \Cref{lem:saturation} is therefore a re-derivation of their worst-case seller. The transform formulation adds a pointwise treatment of atoms: by \cref{eq:Ahat-second}, convexity of $\widehat A$ holds for arbitrary buyer distributions, and it makes the monotonicity of the reduction in \cref{lem:saturation-monotone} immediate.
\end{remark}

The saturated seller is at least as bad for the fixed price as the original one. Let $\widetilde m:=\E[\widetilde S]$ and $\widetilde I:=\E[(B-\widetilde S)_+]$.

\begin{lemma}[Monotonicity of the saturation reduction]\label{lem:saturation-monotone}
For every $h\ge0$, $I-hm\le\widetilde I-h\widetilde m$.
\end{lemma}

\begin{proof}
By \cref{lem:saturation}, $\widetilde A\ge A$ on $[0,p_1]$. For $p\ge p_1$, the function $A$ has slope at most one, so $A(p)\le A(p_1)+p-p_1\le\widehat A(p_1)+p-p_1=\widetilde A(p)$.
Therefore $\widetilde I=\E[\widetilde A(B)]\ge\E[A(B)]=I$. For the seller means,
\[
m=\lim_{p\to\infty}(p-A(p)),\qquad
\widetilde m=\lim_{p\to\infty}(p-\widetilde A(p)),
\]
and $\widetilde A\ge A$ implies $\widetilde m\le m$. Combining the inequalities proves the claim.
\end{proof}

\subsection{Inverse-call coordinates and the control problem}\label{sec:control-coords}

After saturation the seller is a functional of the buyer alone, and the objective $\widetilde I-h\widetilde m$ depends only on the buyer's call transform on $[0,p_1]$. The call transform is the natural clock: we take its value $c=C(p)$ as the independent variable and the buyer survival probability as the control. In these coordinates the objective becomes an integral functional in which the control enters only through itself and its running integral.

Set $c_0:=C(0)>1$. On $[0,p_1]$, use $c=C(p)$ as the independent coordinate. Since $C(p_1)=1>0$, we have $\Pp(B>p)\ge\Pp(B>p_1)>0$ on this interval. Hence $C$ is strictly decreasing there with slope bounded away from zero, and the inverse map $p=p(c)$ is Lipschitz, in particular absolutely continuous. Define $q(c):=-C'(p(c))$ for almost every $c\in[1,c_0]$.
Then $0<q(c)\le1$ and
\begin{equation}\label{eq:change-variable}
\dd p=-\frac{\dd c}{q(c)}.
\end{equation}
\paragraph{The relaxation.}
For controls generated by buyer distributions, $q$ is nondecreasing as a function of $c$. We enlarge the admissible class and retain only the measurable constraint $0<q\le1$. This can only increase the supremum, so the bound obtained remains a valid universal bound. It is the only relaxation in the reduction. A relaxation of the same monotonicity constraint is used by Liu, Ren, and Wang in their first-order analysis~\cite[Theorem~3.4]{LiuRenWang2023}. \Cref{lem:qstar-monotone} shows that the optimizing control is itself nondecreasing, and \cref{sec:extremal} reconstructs it as a genuine buyer distribution, so the relaxation is tight in the limit (\cref{prop:decision-equivalence}).

Define
\begin{equation}\label{eq:Y-R-control}
Y(c):=\int_1^c q(\sigma)\dd\sigma,\qquad
R(c):=\int_c^{c_0}\frac{\dd t}{q(t)t^2}.
\end{equation}
By \cref{eq:change-variable}, $R(C(p))$ coincides with the function $R(p)$ of \cref{eq:R-Ahat}; we use the same letter in both coordinates. By \cref{eq:Ahat-prime}, the saturated seller CDF on $[0,p_1]$, written as a function of $c$, is
\begin{equation}\label{eq:Fhat-control}
\widehat F(c)=\frac1c-q(c)R(c).
\end{equation}

\begin{lemma}[Control representation]\label{lem:control-representation}
For the saturated seller,
\begin{align}
\widetilde m&=\int_1^{c_0}\frac{1-c^{-2}}{q(c)}\dd c,
\label{eq:m-control}\\
\widetilde I&=1+\log c_0-\int_1^{c_0}\frac{Y(c)}{q(c)c^2}\dd c.
\label{eq:I-control}
\end{align}
Consequently, for every $h\ge0$,
\begin{equation}\label{eq:objective-control}
\widetilde I-h\widetilde m=1+\K_h(c_0,q),
\end{equation}
where
\begin{equation}\label{eq:K-definition}
\K_h(b,q):=\log b-\int_1^b\frac{Y(c)+h(c^2-1)}{q(c)c^2}\dd c,
\qquad Y(c)=\int_1^c q(\sigma)\dd\sigma.
\end{equation}
\end{lemma}

\begin{proof}
Since $\widetilde A$ has slope one after $p_1$,
\[
\widetilde m=p_1-\widehat A(p_1)
=\int_1^{c_0}\frac{\dd c}{q(c)}-\int_1^{c_0}\frac{\dd c}{q(c)c^2},
\]
which is \cref{eq:m-control}.

For the first-best gain, use \cref{eq:I-identities}. Beyond $p_1$, the saturated seller CDF equals one, so the tail contribution is
\[
\int_{p_1}^\infty\Pp(B>p)\dd p=C(p_1)=1.
\]
On $[0,p_1]$, changing variables from $p$ to $c$ gives
\[
\int_0^{p_1}\widehat F(p)\Pp(B>p)\dd p=\int_1^{c_0}\widehat F(c)\dd c.
\]
By \cref{eq:Fhat-control},
\begin{align*}
\int_1^{c_0}\widehat F(c)\dd c
&=\log c_0-\int_1^{c_0}q(c)\left(\int_c^{c_0}\frac{\dd t}{q(t)t^2}\right)\dd c\\
&=\log c_0-\int_1^{c_0}\frac{\int_1^t q(c)\dd c}{q(t)t^2}\dd t,
\end{align*}
where Tonelli's theorem justifies exchanging integrals. This is \cref{eq:I-control}. Combining the formulas yields \cref{eq:objective-control}.
\end{proof}

Taking the supremum over all endpoints and relaxed controls defines the value of the control problem. For $h>0$, let
\begin{equation}\label{eq:V-definition}
V(h):=\sup_{b>1}\sup_{q\in\Qcal_b}\K_h(b,q),
\end{equation}
where $\Qcal_b$ is the set of measurable functions $q:(1,b)\to(0,1]$ for which the integral in \cref{eq:K-definition} is finite.

\begin{corollary}[Universal reduction]\label{cor:universal-reduction}
Under the normalization $M=1$, $I-hm\le1+V(h)$ for every $h>0$.
\end{corollary}

\begin{proof}
If $C(0)\le1$, then $I-hm\le1\le1+V(h)$ because $V(h)\ge0$; indeed $V(h)>0$, as shown in the proof of \cref{lem:free-existence}. If $C(0)>1$, combine \cref{lem:saturation-monotone,lem:control-representation} and the definition of $V(h)$.
\end{proof}

\section{Global solution of the variational problem}\label{sec:variational}

This section computes $V(h)$ exactly for every $h>0$. The argument proceeds in five steps: existence of optimizers (\cref{sec:existence}); a convex reformulation whose KKT conditions certify global optimality (\cref{sec:convexity}); regularity and the one-switch structure (\cref{sec:one-switch}); the free-endpoint condition (\cref{sec:terminal}); and explicit integration of the optimality system (\cref{sec:integration,sec:exact-value}).

For fixed $h>0$ and endpoint $b>1$, write
\begin{equation}\label{eq:N-definition}
N(c):=Y(c)+h(c^2-1),\qquad Y(c)=\int_1^c q(\sigma)\dd\sigma.
\end{equation}
Maximizing $\K_h(b,q)$ is equivalent to minimizing
\begin{equation}\label{eq:A-action}
\Acal_{h,b}(q):=\int_1^b\frac{N(c)}{q(c)c^2}\dd c.
\end{equation}

\subsection{Existence}\label{sec:existence}

\begin{proposition}[Existence of optimizers]\label{prop:optimizer-existence}
For every $h>0$ and $b>1$, the fixed-endpoint action $\Acal_{h,b}$ has a minimizer. The free-endpoint supremum $V(h)$ is attained by some pair $(b,q)$.
\end{proposition}

We use the direct method in the calculus of variations~\cite{Dacorogna2008}. The difficulty is that the natural compact closure of the admissible controls permits $q=0$, where the action has a reciprocal singularity. Lower semicontinuity of the reciprocal integrand follows from an integral form of the perspective identity $a/v=\sup_{\theta\ge0}(2\sqrt{a\theta}-\theta v)$.

\begin{claim}[Integral perspective representation]\label{clm:perspective-integral}
Let $a,v$ be nonnegative measurable functions on a finite interval, and use the convention $a/v=+\infty$ on $\{a>0,v=0\}$. Then
\begin{equation}\label{eq:perspective-dual}
\int\frac{a}{v}
=\sup_{\theta\in L^\infty_+}
\int\left(2\sqrt{a\theta}-\theta v\right),
\end{equation}
where the supremum is over bounded nonnegative measurable functions.
\end{claim}

\begin{proof}
Pointwise,
\[
\frac{a}{v}=\sup_{\theta\ge0}\left(2\sqrt{a\theta}-\theta v\right).
\]
For $v>0$, the maximizer is $\theta=a/v^2$; for $a>0,v=0$, the supremum is infinite. Define
\[
\theta_n=
\begin{cases}
\min\{n,a/v^2\},&v>0,\\
n,&v=0,
\end{cases}
\]
with $\theta_n=0$ on $\{a=v=0\}$. Then the corresponding nonnegative integrands increase pointwise to $a/v$. Monotone convergence proves the reverse inequality in \cref{eq:perspective-dual}; the forward inequality follows from the pointwise bound.
\end{proof}

\begin{lemma}[Fixed-endpoint existence]\label{lem:fixed-existence}
For every $h>0$ and $b>1$, the minimum of $\Acal_{h,b}$ over $\Qcal_b$ is attained.
\end{lemma}

\begin{proof}
Introduce the weak-star compact closure
\[
\Ycal_b:=\{Y\in W^{1,\infty}(1,b):Y(1)=0,\ 0\le Y'\le1\text{ a.e.}\}.
\]
For $c>1$, set
\[
\omega(c,Y):=\frac{Y(c)+h(c^2-1)}{c^2}>0.
\]
The extended action is $\int \omega(c,Y)/Y'(c)\dd c$, with value $+\infty$ wherever $Y'=0$ and $\omega>0$.

Let $(Y_n)$ be a minimizing sequence. By the Arzel\`a--Ascoli theorem and weak-star compactness of the unit ball of $L^\infty$, after passing to a subsequence,
\[
Y_n\to Y\ \text{uniformly},\qquad
Y_n'\stackrel{*}{\rightharpoonup}Y'\ \text{in }L^\infty,
\]
for some $Y\in\Ycal_b$. The functions $\omega(\cdot,Y_n)$ converge uniformly to $\omega(\cdot,Y)$. For every bounded nonnegative measurable $\theta$, dominated convergence in the square-root term and weak-star convergence in the linear term yield
\begin{align*}
&\int_1^b\left(2\sqrt{\omega(c,Y_n)\theta(c)}-\theta(c)Y_n'(c)\right)\dd c\\
&\hspace{3cm}\longrightarrow
\int_1^b\left(2\sqrt{\omega(c,Y)\theta(c)}-\theta(c)Y'(c)\right)\dd c.
\end{align*}
Taking the supremum over bounded $\theta$ and using \cref{clm:perspective-integral} gives
\[
\Acal_{h,b}(Y')\le\liminf_{n\to\infty}\Acal_{h,b}(Y_n').
\]
Thus $Y$ minimizes the extended action. Since $\omega(c,Y)>0$ for every $c>1$, finite action forces $Y'(c)>0$ almost everywhere; hence $q=Y'$ belongs to $\Qcal_b$.
\end{proof}

\begin{lemma}[Free-endpoint existence]\label{lem:free-existence}
For every $h>0$, the supremum defining $V(h)$ is attained by some pair $(b,q)$ with $b>1$ and $q\in\Qcal_b$.
\end{lemma}

\begin{proof}
Since $q\le1$ and $Y\ge0$,
\begin{equation}\label{eq:endpoint-coercive}
\K_h(b,q)\le\log b-h\int_1^b(1-c^{-2})\dd c
=\log b-h(b+b^{-1}-2).
\end{equation}
The right-hand side tends to $-\infty$ as $b\to\infty$. No positive maximizing sequence can have $b\downarrow1$ because $\K_h(b,q)\le\log b$. A positive value does exist: for $q\equiv1$, the running reward
\[
\frac1c-\frac{c-1+h(c^2-1)}{c^2}
\]
tends to one as $c\downarrow1$, and is therefore positive on a sufficiently short interval.

It remains to treat endpoints in a compact subinterval of $(1,\infty)$. Let $b_n\to b>1$, put $\rho_n=b_n-1$, and reparameterize
\[
Z_n(x):=\frac{Y_n(1+\rho_nx)}{\rho_n},\qquad0\le x\le1.
\]
Then $Z_n(0)=0$ and $0\le Z_n'\le1$. After a subsequence, $Z_n\to Z$ uniformly and $Z_n'\stackrel{*}{\rightharpoonup}Z'$ in $L^\infty(0,1)$. The action is
\begin{equation}\label{eq:rescaled-action}
\Acal_{h,b_n}(q_n)
=\rho_n\int_0^1
\frac{\rho_nZ_n(x)+h((1+\rho_nx)^2-1)}{Z_n'(x)(1+\rho_nx)^2}\dd x.
\end{equation}
The numerator coefficient converges uniformly, while $Z_n'\stackrel{*}{\rightharpoonup}Z'$ in $L^\infty(0,1)$. Applying \cref{clm:perspective-integral} on $[0,1]$ proves lower semicontinuity of \cref{eq:rescaled-action}. Since $\log b_n\to\log b$, the objective is upper semicontinuous along a maximizing sequence, and the limit pair attains $V(h)$.
\end{proof}

\begin{proof}[Proof of \cref{prop:optimizer-existence}]
The two assertions are \cref{lem:fixed-existence,lem:free-existence}.
\end{proof}

\subsection{Logarithmic convexity and global sufficiency}\label{sec:convexity}

The action $\Acal_{h,b}$ is not convex in $q$. The key observation is that it becomes strictly convex in the logarithm of the control. Fix $b>1$ and write $u(c):=\log q(c)\le0$; we write $\Acal_{h,b}(u)$ for $\Acal_{h,b}(e^u)$. Tonelli's theorem rewrites the action as
\begin{align}
\Acal_{h,b}(u)
={}&h\int_1^b(1-c^{-2})e^{-u(c)}\dd c
\label{eq:log-action-first}\\
&+\int_1^b\int_1^c\frac{e^{u(\sigma)-u(c)}}{c^2}\dd\sigma\dd c.
\label{eq:log-action-second}
\end{align}

\begin{proposition}[Strict logarithmic convexity]\label{prop:strict-convexity}
For fixed $h>0$ and $b>1$, $\Acal_{h,b}(u)$ is strictly convex on
\[
\Ucal_b:=\{u:(1,b)\to(-\infty,0]\text{ measurable}:\Acal_{h,b}(u)<\infty\}.
\]
The domain $\Ucal_b$ is convex, because convexity of the integrands in \cref{eq:log-action-first,eq:log-action-second} preserves finiteness along line segments. Consequently, the fixed-endpoint minimizing control is unique up to null sets.
\end{proposition}

\begin{proof}
Each integrand in \cref{eq:log-action-first,eq:log-action-second} is the exponential of a linear functional of $u$, and is therefore convex. If two controls differ on a set of positive measure, strict convexity of $x\mapsto e^{-x}$ and the positive weight $h(1-c^{-2})$ make the first integral strictly convex. Hence equality in the convexity inequality is possible only for controls that agree almost everywhere.
\end{proof}

For a fixed-endpoint minimizer $q=e^u$, define
\begin{equation}\label{eq:lambda}
\lambda(c):=\int_c^b\frac{\dd t}{q(t)t^2},\qquad c>1.
\end{equation}
It is finite on every compact subinterval of $(1,b]$. The first-variation density of the action is
\begin{equation}\label{eq:g-gradient}
\Gamma(c):=q(c)\lambda(c)-\frac{N(c)}{q(c)c^2}.
\end{equation}
Because of strict convexity, the following conditions are not only necessary but also sufficient for global optimality; see~\cite{Rockafellar1970} for the general principle.

\begin{proposition}[KKT conditions and global sufficiency]\label{prop:global-kkt}
A fixed-endpoint minimizer satisfies
\begin{align}
\Gamma(c)&=0&&\text{for almost every }c\text{ with }q(c)<1,\label{eq:kkt-interior}\\
\Gamma(c)&\le0&&\text{for almost every }c\text{ with }q(c)=1.\label{eq:kkt-boundary}
\end{align}
Equivalently,
\begin{equation}\label{eq:q-kkt}
q(c)=\min\left\{1,\sqrt{\frac{N(c)}{\lambda(c)c^2}}\right\}
\quad\text{almost everywhere}.
\end{equation}
Conversely, any feasible control satisfying \cref{eq:kkt-interior,eq:kkt-boundary} is a global fixed-endpoint minimizer.
\end{proposition}

\begin{proof}
\emph{Necessary conditions.}
Let $\varphi$ be bounded and supported in $[1+\eta,b-\eta]$ for some $\eta>0$, and put $u_\tau=u+\tau\varphi$. For $|\tau|$ small, each integrand in \cref{eq:log-action-first,eq:log-action-second} is bounded by a constant multiple of the corresponding integrand at $u$. Differentiation under the integral sign and Fubini's theorem are therefore justified. The derivative at zero is
\begin{equation}\label{eq:first-variation}
\left.\frac{\dd}{\dd\tau}\Acal_{h,b}(u_\tau)\right|_{\tau=0}
=\int_1^b\Gamma(c)\varphi(c)\dd c.
\end{equation}
Indeed, the occurrences of $u(c)$ as the inner variable contribute $q(c)\lambda(c)$, while the first term and the occurrences as the outer variable contribute $-N(c)/(q(c)c^2)$.

On $\{q\le1-\delta\}\cap[1+\eta,b-\eta]$, both signs of sufficiently small perturbations are feasible, so \cref{eq:first-variation} implies $\Gamma=0$ there. Taking the union over rational $\delta,\eta>0$ proves \cref{eq:kkt-interior}. On the active set $\{q=1\}\cap[1+\eta,b-\eta]$, every bounded $\varphi\le0$ is a feasible one-sided direction; minimality gives $\int \Gamma\varphi\ge0$, and hence $\Gamma\le0$ almost everywhere there. Letting $\eta\downarrow0$ proves \cref{eq:kkt-boundary}. Algebraically these two conditions are exactly \cref{eq:q-kkt}.

\emph{Global sufficiency.}
Let $\bar u\le0$ be any other feasible logarithmic control. If $\psi:=\bar u-u$ is bounded, apply $e^x\ge1+x$ separately to every exponential in \cref{eq:log-action-first,eq:log-action-second}. After collecting linear terms one obtains the global convexity inequality
\begin{equation}\label{eq:global-sufficiency}
\Acal_{h,b}(\bar u)-\Acal_{h,b}(u)
\ge\int_1^b\Gamma(c)(\bar u(c)-u(c))\dd c.
\end{equation}
(The linear terms are integrable: $\psi$ is bounded, $N/(qc^2)$ is integrable by finiteness of the action, and $\int_1^b q\lambda\le\int_1^b(t-1)/(q(t)t^2)\dd t<\infty$ by Tonelli's theorem.) On $\{q<1\}$ the integrand on the right is zero. On $\{q=1\}$, one has $u=0$, $\bar u\le0$, and $\Gamma\le0$, so the integrand is nonnegative. Thus $\Acal_{h,b}(\bar u)\ge\Acal_{h,b}(u)$.

For an unbounded log-ratio, let $\psi_n=\max\{-n,\min\{n,\psi\}\}$ and $u_n=u+\psi_n$. Pointwise, $u_n$ lies between $u$ and $\bar u$, so $u_n\le0$. The bounded case gives $\Acal_{h,b}(u_n)\ge\Acal_{h,b}(u)$. Moreover, $e^{-u_n(c)}\le e^{-u(c)}+e^{-\bar u(c)}$
and, since $u_n(\sigma)\le0$,
\[
\int_1^c\frac{e^{u_n(\sigma)-u_n(c)}}{c^2}\dd\sigma
\le\frac{c-1}{c^2}\bigl(e^{-u(c)}+e^{-\bar u(c)}\bigr).
\]
The right-hand sides are integrable because $(c-1)/c^2\le1-c^{-2}$ and both controls have finite action. Dominated convergence gives $\Acal_{h,b}(u_n)\to\Acal_{h,b}(\bar u)$, completing the proof.
\end{proof}

\subsection{Regularity and the one-switch structure}\label{sec:one-switch}

The KKT conditions express the control pointwise through the sign of $N-\lambda c^2$. We now show that this switching function can change sign only once, which fixes the qualitative shape of every optimizer.

\begin{figure}[t]
\centering
% Figure: the one-switch structure of an optimal control q, drawn at h = h_* (exact curves;
% data from gen_figure_data.py). Labels use the generic control q; the text calls this control q_*.
% Horizontal coordinate in both panels: c - 1 (so c = 1 is at the origin).
% Constants (r = r_*): s = 1.335116, b = 2.055592.
\begin{tikzpicture}
  \def\sc{0.335116}  % s - 1
  \def\bc{1.055592}  % b - 1
  % ---------- panel (a): q_*(c) ----------
  \begin{scope}[x=6.0cm, y=3.3cm]
    \draw[fp guide] (0,1) -- (\sc,1);
    \draw[fp guide] (\sc,0) -- (\sc,1);
    \draw[fp guide] (\bc,0) -- (\bc,1);
    \draw[fp accent] plot file {figures/data/qstar_interior.dat};
    \draw[fp main] (\sc,1) -- (\bc,1);
    \draw[fp axis] (0,0) -- (1.15,0) node[right, fp label] {$c$};
    \draw[fp axis] (0,0) -- (0,1.22) node[above, fp label] {$q(c)$};
    \draw[fp tick] (0,1) -- (-0.02,1) node[left, fp label] {$1$};
    \node[below, fp label] at (0,0) {$1$};
    \draw[fp tick] (\sc,0) -- (\sc,-0.035) node[below, fp label] {$s$};
    \draw[fp tick] (\bc,0) -- (\bc,-0.035) node[below, fp label] {$b$};
    \node[fp atom, fill=fpblue] at (0,0) {};
    \node[fp atom] at (\sc,1) {};
    \node[fp note, text=fpblue, anchor=west] at (0.125,0.24) {interior arc\\$0<q<1$};
    \node[fp note, anchor=south] at ({(\sc+\bc)/2},1.0) {boundary arc $q=1$};
    \node[fp panel] at (0.55,-0.30) {(a) optimal control};
  \end{scope}
  % ---------- panel (b): switching function Delta(c) ----------
  \begin{scope}[shift={(8.0cm,1.19cm)}, x=6.0cm, y=0.95cm]
    \fill[fp region] (0,-1.25) rectangle (\sc,0);
    \draw[fp guide] (\sc,-1.25) -- (\sc,2.2);
    \draw[fp guide] (\bc,0) -- (\bc,2.0556);
    \draw[fp accent] plot file {figures/data/delta_interior.dat};
    \draw[fp main] plot file {figures/data/delta_boundary.dat};
    \draw[fp axis] (0,0) -- (1.15,0) node[right, fp label] {$c$};
    \draw[fp axis] (0,-1.25) -- (0,2.4) node[above, fp label] {$\Delta(c)$};
    \node[left, fp label] at (0,0) {$0$};
    \node[fp atom] at (\sc,0) {};
    \node[fp label, anchor=south east] at (\sc,0.02) {$s$};
    \draw[fp tick] (\bc,0) -- (\bc,-0.1) node[below, fp label] {$b$};
    \node[fp note, anchor=south] at ({\sc/2},-1.22) {$\Delta<0$\\$q<1$};
    \node[fp note, anchor=north] at ({(\sc+\bc)/2+0.06},-0.15) {$\Delta\ge0$, \ $q=1$};
    \node[fp panel] at (0.55,-2.297) {(b) switching function};
  \end{scope}
\end{tikzpicture}
\caption{The one-switch structure of an optimal control (\cref{lem:one-switch,lem:trajectory-reconstruction}). The curves are drawn for the parameter value $h=h_\star$ fixed in \cref{sec:fixed-point} and computed from the closed-form trajectory of \cref{lem:trajectory-reconstruction}; this particular control is denoted $q_\star$ in \cref{sec:extremal}. Horizontal axes are in units of $c$, starting at $c=1$. (a)~The control rises from $q(1)=0$ along the interior arc (blue), reaches $q(s)=1$ at the unique switch $s$, and stays on the upper boundary $q=1$ (black) until the endpoint $b$. (b)~The switching function $\Delta=N-\lambda c^2$ of \cref{lem:switch-regularity} is negative before $s$, where $q<1$, and nonnegative after it, where $q=1$; it has a single upward crossing at $s$. The computed curve happens to be increasing, but \cref{lem:one-switch} uses only the crossing property.}\label{fig:one-switch}
\end{figure}

\begin{lemma}[Regularity of the switching function]\label{lem:switch-regularity}
A minimizing control has a representative for which $q$ is continuous and strictly positive on $(1,b)$, $N,\lambda\in C^1(1,b)$, and
\[
N'=q+2hc,\qquad \lambda'=-\frac1{qc^2}.
\]
Consequently, $\Delta(c):=N(c)-\lambda(c)c^2$
belongs to $C^1(1,b)$.
\end{lemma}

\begin{proof}
The functions $N$ and $\lambda$ are continuous and positive on $(1,b)$ before redefining the control on a null set. The right-hand side of \cref{eq:q-kkt} is therefore a continuous, strictly positive function on $(1,b)$. Choose that continuous representative of $q$. Since it agrees with the original control almost everywhere, it does not change $Y$ or the action. The displayed derivative identities then imply the asserted $C^1$ regularity.
\end{proof}

We call $\Delta$ the \emph{switching function}. It carries the same sign information as the first-variation density $\Gamma$ of \cref{eq:g-gradient}, but in a form whose zeros can be analyzed. On the active set $\{q=1\}$ one has $\Gamma=\lambda-N/c^2=-\Delta/c^2$, so the boundary condition \cref{eq:kkt-boundary} is $\Delta\ge0$ there. On $\{q<1\}$ the interior condition \cref{eq:kkt-interior} gives $q^2=N/(\lambda c^2)$, which is less than one exactly when $\Delta<0$. For the continuous representative of \cref{lem:switch-regularity}, \cref{eq:q-kkt} holds at every $c\in(1,b)$, and therefore
\begin{equation}\label{eq:Delta-sign}
q(c)<1\iff\Delta(c)<0,
\qquad
q(c)=1\iff\Delta(c)\ge0
\qquad(1<c<b).
\end{equation}

\begin{lemma}[Unique switch]\label{lem:one-switch}
Every free-endpoint maximizer has a unique $s\in(1,b)$ such that
\begin{equation}\label{eq:one-switch}
q(c)<1\quad(1<c<s),\qquad q(c)=1\quad(s\le c\le b).
\end{equation}
\end{lemma}

\begin{proof}
At every zero of $\Delta$, \cref{eq:q-kkt} gives $q=1$ and $\lambda=N/c^2$. Hence
\begin{align}
\Delta'(c)
&=q(c)+2hc+\frac1{q(c)}-2\lambda(c)c\notag\\
&=\frac{2(c+h-Y(c))}{c}
\ge\frac{2(1+h)}c>0,
\label{eq:Delta-positive-crossing}
\end{align}
where the inequality uses $Y(c)\le c-1$. Thus every zero is a strict crossing from negative to positive.

Near $c=1$, $N(c)\to0$, while $\lambda$ is bounded below by its positive integral over any fixed terminal subinterval, so $\Delta<0$. Near $b$, $\lambda(c)\to0$ and $N(b)>0$, so $\Delta>0$. A zero therefore exists. If there were a second zero, the first return to zero after the initial crossing would have a nonpositive left derivative, contradicting \cref{eq:Delta-positive-crossing} and the $C^1$ regularity from \cref{lem:switch-regularity}. The sign rule \cref{eq:Delta-sign} gives \cref{eq:one-switch}.
\end{proof}

\Cref{fig:one-switch} shows this structure for the optimizer at $h=h_\star$, drawn from the explicit solution obtained in \cref{sec:integration}. The lemma itself uses only two facts: $\Delta$ is negative near $c=1$ and positive near $b$, and every zero of $\Delta$ is an upward crossing.

\begin{lemma}[Endpoint regularity]\label{lem:endpoint-regularity}
On the interior arc $(1,s)$,
\begin{equation}\label{eq:interior-lambda-relation}
\lambda(c)=\frac{N(c)}{q(c)^2c^2}.
\end{equation}
Moreover,
\[
0<\lambda(1+)<\infty,\qquad
q(c)\to0\quad(c\downarrow1),\qquad
\int_1^s\frac{\dd c}{q(c)}<\infty.
\]
\end{lemma}

\begin{proof}
\Cref{eq:interior-lambda-relation} is the equality case of \cref{eq:q-kkt}. Differentiating its square root through $\lambda'=-1/(qc^2)$ gives
\begin{equation}\label{eq:sqrt-lambda}
\frac{\dd}{\dd c}\sqrt{\lambda(c)}=-\frac1{2c\sqrt{N(c)}}.
\end{equation}
Since $N(c)\ge h(c^2-1)$,
\[
\int_1^s\frac{\dd c}{c\sqrt{N(c)}}
\le\frac1{\sqrt h}\int_1^s\frac{\dd c}{c\sqrt{c^2-1}}<\infty.
\]
Thus $0<\lambda(1+)<\infty$. Because $N(c)\to0$, \cref{eq:interior-lambda-relation} gives $q(c)\to0$. Finally,
\[
\frac1{q(c)}=\frac{c\sqrt{\lambda(c)}}{\sqrt{N(c)}}
\le\frac{K_0}{\sqrt{c^2-1}}
\]
near one, for a constant $K_0$, which is integrable.
\end{proof}

\subsection{The free-endpoint equation}\label{sec:terminal}

The endpoint $b$ is also a decision variable. Optimality with respect to $b$ gives one scalar terminal condition.

\begin{lemma}[Terminal condition]\label{lem:terminal-condition}
Every free-endpoint maximizer satisfies $N(b)=b$.
\end{lemma}

\begin{proof}
Let
\[
\ell(c):=\frac1c-\frac{N(c)}{q(c)c^2}
\]
be the running reward. By \cref{lem:one-switch}, $q=1$ on a left neighborhood of $b$, and $\ell$ has a continuous left limit there. For small $\delta>0$, maximality against truncation gives
\begin{equation}\label{eq:terminal-truncate}
0\le \K_h(b,q)-\K_h(b-\delta,q|_{(1,b-\delta)})
=\int_{b-\delta}^b\ell(c)\dd c.
\end{equation}
Hence $\ell(b-)\ge0$.

For extension, set $\widetilde q=1$ on $[b,b+\delta]$ and define
\[
\widetilde Y(c)=Y(b)+c-b,
\qquad
\widetilde N(c)=\widetilde Y(c)+h(c^2-1).
\]
Then maximality gives
\begin{equation}\label{eq:terminal-extend}
0\ge \K_h(b+\delta,\widetilde q)-\K_h(b,q)
=\int_b^{b+\delta}
\left(\frac1c-\frac{\widetilde N(c)}{c^2}\right)\dd c.
\end{equation}
Divide \cref{eq:terminal-truncate,eq:terminal-extend} by $\delta$ and let $\delta\downarrow0$. The two continuous one-sided limits agree, so $\ell(b)=0$. Since $q(b)=1$, this is exactly $N(b)=b$.
\end{proof}

\subsection{Integration of the interior trajectory}\label{sec:integration}

By \cref{lem:switch-regularity} and the interior equality \cref{eq:interior-lambda-relation}, the control is $C^1$ on $(1,s)$. Differentiating that equality and using $\lambda'=-1/(qc^2)$ gives
\begin{equation}\label{eq:q-prime}
q'=\frac{q(cq+h-Y)}{cN}.
\end{equation}
The following normalization isolates the integrable part of the dynamics. The state $y=N/(hc^2)$ removes the deterministic quadratic drift in $N$; the variable $v=q/(hc)$ puts the control on the same scale; and $z=v/y=qc/N$ is the ratio appearing in the running reward. Set
\begin{equation}\label{eq:y-v-z}
t:=\log c,\qquad y:=\frac{N}{hc^2},\qquad v:=\frac{q}{hc},\qquad z:=\frac vy.
\end{equation}
Substitution of $N'=q+2hc$ and \cref{eq:q-prime} gives
\begin{align}
\dot y&=v+2-2y,\label{eq:y-dot}\\
\dot v&=\frac{v(v+1-2y)}y,\label{eq:v-dot}\\
\dot z&=-\frac zy.\label{eq:z-dot}
\end{align}
Finally, $w=1/y$ makes the equation in $z$ linear. Eliminating $t$ yields
\[
\frac{\dd w}{\dd z}=1+\frac{2(w-1)}z,
\qquad
\frac{\dd}{\dd z}\left(\frac{w}{z^2}\right)
=\frac1{z^2}-\frac2{z^3}.
\]
Integration gives $w/z^2=-1/z+1/z^2+\kappa$. Therefore every interior trajectory has the first integral
\begin{equation}\label{eq:first-integral}
\frac1y=1-z+\kappa z^2.
\end{equation}
To identify the constant $\kappa$, note that
\[
\frac{1/y}{z^2}=\frac{y}{v^2}=h\lambda(c)c^2,
\]
so \cref{lem:endpoint-regularity} implies $\kappa=h\lambda(1+)>0$. Also
\begin{equation}\label{eq:z-infinity}
z=\frac{qc}{N}=\frac1{\sqrt{\lambda N}}\longrightarrow\infty
\qquad(c\downarrow1),
\end{equation}
and hence
\begin{equation}\label{eq:log-c-z}
\log c=\int_z^\infty\frac{\dd\zeta}{\zeta(1-\zeta+\kappa \zeta^2)}.
\end{equation}

Let
\begin{equation}\label{eq:r-definition}
r:=\frac bs>1,\qquad a_r:=\frac{r-1}{r}.
\end{equation}
On the terminal arc $q=1$,
\begin{equation}\label{eq:lambda-s}
\lambda(s)=\int_s^b\frac{\dd c}{c^2}=\frac1s-\frac1b.
\end{equation}
At the switch $\Delta(s)=0$, so
\begin{equation}\label{eq:N-s}
N(s)=s^2\lambda(s)=s\left(1-\frac1r\right)=a_rs.
\end{equation}
Integrating $N'=1+2hc$ from $s$ to $b$ and using $N(b)=b$ yields
\begin{equation}\label{eq:hs-r}
hs=\frac1{r(r^2-1)}.
\end{equation}
At the switch,
\begin{equation}\label{eq:switch-y-v-z}
y_s=(r-1)(r^2-1),\qquad
v_s=r(r^2-1),\qquad
z_s=\frac r{r-1}.
\end{equation}
Substitution in \cref{eq:first-integral} gives
\begin{equation}\label{eq:kappa-r}
\kappa=\frac1{r+1}.
\end{equation}

Set $x:=1/z$. Then
\begin{equation}\label{eq:y-v-x}
y=\frac{x^2}{d_r(x)},\qquad
v=\frac{x}{d_r(x)},\qquad
d_r(x):=x^2-x+\frac1{r+1},
\end{equation}
and
\begin{equation}\label{eq:log-c-x}
\frac{\dd\log c}{\dd x}=\frac{x}{d_r(x)}.
\end{equation}
The variable $x$ runs from $0$ at $c=1$ to $a_r$ at $c=s$. Positivity of $y$ requires $d_r>0$ throughout this interval. This forces $1<r<3$: if $r\ge3$, then $d_r(1/2)\le0$ and $1/2<a_r$.

Define, for $1<r<3$,
\begin{align}
T(r)&:=\int_0^{a_r}\frac{x}{d_r(x)}\dd x,\label{eq:T-definition}\\
J(r)&:=\int_0^{a_r}\frac{x(1-x)}{d_r(x)}\dd x.\label{eq:J-definition}
\end{align}
Then $s=e^{T(r)}$, and \cref{eq:hs-r} becomes
\begin{equation}\label{eq:Hminus}
h=H_-(r):=\frac{e^{-T(r)}}{r(r^2-1)}.
\end{equation}

\begin{lemma}[Parameter uniqueness]\label{lem:Hminus-bijection}
The function $H_-$ is continuous and strictly decreasing from $+\infty$ to $0$ on $(1,3)$. Hence, for every $h>0$, there is a unique $r(h)\in(1,3)$ satisfying $h=H_-(r(h))$.
\end{lemma}

\begin{proof}
As $r$ increases, the endpoint $a_r$ increases and the constant term $1/(r+1)$ in $d_r$ decreases. Thus the positive integrand and its interval both increase, so $T$ is strictly increasing. The factor $r(r^2-1)$ is also strictly increasing, and therefore $H_-$ is strictly decreasing. As $r\downarrow1$, $T(r)\to0$ and $r(r^2-1)\to0$, so $H_-(r)\to\infty$. As $r\uparrow3$, $d_r$ develops a double zero at $x=1/2$ in the integration interval, so $T(r)\to\infty$ and $H_-(r)\to0$.
\end{proof}

\begin{lemma}[Unique trajectory reconstruction]\label{lem:trajectory-reconstruction}
Fix $h>0$ and let $r=r(h)$. Put $a=a_r$, $d=d_r$, $s=e^{T(r)}$, and $b=rs$.
There is exactly one one-switch trajectory satisfying the KKT equations and the terminal condition. It is given on the interior arc by
\begin{align}
c(x)&=\exp\left(\int_0^x\frac{\xi}{d(\xi)}\dd\xi\right),\label{eq:c-x-reconstruction}\\
q(c(x))&=hc(x)\frac{x}{d(x)},\qquad0<x\le a,\label{eq:q-x-reconstruction}
\end{align}
and by $q(c)=1$ for $s\le c\le b$. The control has finite action, is continuous, and satisfies $q(c)\downarrow0$ as $c\downarrow1$.
\end{lemma}

\begin{proof}
The identities derived above force $s=e^{T(r)}$, $b=rs$, and \cref{eq:c-x-reconstruction,eq:q-x-reconstruction}; hence at most one trajectory is possible. We verify that these formulas reconstruct one.

Since $d>0$ on $[0,a]$, $c$ is strictly increasing from $1$ to $s$. The switch relation $hs=1/[r(r^2-1)]$ and
\[
d(a)=\frac1{r^2(r+1)}
\]
give $q(s)=1$. On the interior arc,
\begin{equation}\label{eq:q-reconstruction-monotone}
\frac{\dd}{\dd x}\log q(c(x))
=\frac1x+\frac{1-x}{d(x)}>0,
\end{equation}
so $0<q<1$ on the interior arc. Near zero, $q(c(x))=O(x)$ and $\dd c=O(x)\dd x$; in fact $\dd c/q=\dd x/h$, proving finite action.

Define on the interior arc
\[
N(c(x)):=hc(x)^2\frac{x^2}{d(x)},
\qquad
Y(c):=N(c)-h(c^2-1).
\]
A direct differentiation using $\dd c/\dd x=cx/d$ gives $\dd Y/\dd x=q\,\dd c/\dd x$, and $Y(1)=0$. Thus $Y'=q$. Set
\[
\lambda(c(x)):=\frac{N(c(x))}{q(c(x))^2c(x)^2}
=\frac{d(x)}{hc(x)^2}.
\]
Differentiation gives $\lambda'=-1/(qc^2)$. At $c=s$, this equals $1/s-1/b$, so it matches the terminal definition $\lambda(c)=1/c-1/b$. On $[s,b]$ the state is
\[
N(c)=c-\frac sr+h(c^2-s^2).
\]
It satisfies $N(b)=b$, and
\[
\Delta(c)=N(c)-\lambda(c)c^2
=-\frac sr+h(c^2-s^2)+\frac{c^2}{b}\ge0,
\]
with equality only at $s$. Hence the reconstructed control satisfies the global KKT conditions, the one-switch sign conditions, and the terminal equation. Uniqueness follows from the forced formulas.
\end{proof}

\subsection{Exact value of the control problem}\label{sec:exact-value}

On the interior arc, $x=N/(qc)$, so the running reward is $(1-x)/c$. By \cref{eq:log-c-x},
\begin{equation}\label{eq:K-interior}
\K_{\mathrm{int}}=\int_0^{a_r}\frac{x(1-x)}{d_r(x)}\dd x=J(r).
\end{equation}
On the terminal arc, $q=1$ and $N(c)=c-s/r+h(c^2-s^2)$. Therefore
\begin{align}
\K_{\mathrm{bd}}
&=\int_s^{rs}\left(\frac{s}{rc^2}-h\left(1-\frac{s^2}{c^2}\right)\right)\dd c\notag\\
&=\frac{r-1}{r^2}-hs\frac{(r-1)^2}{r}
=\frac{r-1}{r(r+1)},
\label{eq:K-boundary}
\end{align}
where the last equality uses \cref{eq:hs-r}.

\begin{theorem}[Exact variational value]\label{thm:V-exact}
For every $h>0$, let $r(h)\in(1,3)$ be the unique solution of
\[
h=\frac{e^{-T(r)}}{r(r^2-1)}.
\]
Then
\begin{equation}\label{eq:V-Hplus}
V(h)=H_+(r(h)),
\end{equation}
where
\begin{equation}\label{eq:Hplus}
H_+(r):=J(r)+\frac{r-1}{r(r+1)}.
\end{equation}
The free-endpoint maximizer is unique, and for each fixed endpoint the minimizing control is unique almost everywhere.
\end{theorem}

\begin{proof}
By \cref{lem:free-existence}, a free-endpoint maximizer exists. Every maximizer satisfies the global KKT conditions, the unique-switch lemma, and the terminal condition. Hence \cref{lem:Hminus-bijection} fixes its parameter to $r(h)$, while \cref{lem:trajectory-reconstruction} uniquely reconstructs its endpoint and control. Its value is the sum of \cref{eq:K-interior,eq:K-boundary}. This proves \cref{eq:V-Hplus} for the full measurable-control problem. Fixed-endpoint uniqueness is \cref{prop:strict-convexity}.
\end{proof}

\section{The exact approximation ratio}\label{sec:fixed-point}

\Cref{thm:V-exact} expresses $V(h)$ through the parameter $r(h)$. By \cref{cor:universal-reduction}, the best guarantee obtainable from the reduction is $(1+h_\star)^{-1}$, where $h_\star$ is the smallest $h$ with $V(h)\le h$. This section rewrites both $h$ and $V(h)$ in terms of the single integral $\Theta$, shows that $V(h)=h$ has a unique solution, and derives the universal guarantee. The matching upper bound is the subject of \cref{sec:extremal}.

\subsection{The scalar root equation}

Recall the function $\Theta$ of \cref{eq:theta-intro}. In the notation $a_r=(r-1)/r$ and $d_r(x)=x^2-x+1/(r+1)$ of \cref{eq:r-definition,eq:y-v-x}, it reads $\Theta(r)=\int_0^{a_r}\dd x/d_r(x)$ for $1<r<3$. Since $x=(2x-1)/2+1/2$ and $d_r(a_r)/d_r(0)=1/r^2$,
\begin{equation}\label{eq:T-Theta}
T(r)=\frac{\Theta(r)}2-\log r.
\end{equation}
Also $x(1-x)=1/(r+1)-d_r(x)$, so
\begin{equation}\label{eq:J-Theta}
J(r)=\frac{\Theta(r)}{r+1}-\frac{r-1}{r}.
\end{equation}
Consequently,
\begin{align}
H_-(r)&=\frac{e^{-\Theta(r)/2}}{r^2-1},\label{eq:Hminus-Theta}\\
H_+(r)&=\frac{\Theta(r)-r+1}{r+1}.\label{eq:Hplus-Theta}
\end{align}

\begin{lemma}[Unique fixed point]\label{lem:unique-fixed-point}
There is a unique $r_\star\in(1,3)$ such that $H_-(r_\star)=H_+(r_\star)$. Equivalently, $r_\star$ is the unique solution of
\begin{equation}\label{eq:root-equation}
e^{-\Theta(r)/2}=(r-1)(\Theta(r)-r+1).
\end{equation}
\end{lemma}

\begin{proof}
By \cref{lem:Hminus-bijection}, $H_-$ is strictly decreasing from $+\infty$ to $0$. We prove that $H_+$ is strictly increasing from $0$ to $+\infty$.

Let
\[
f_r(x):=\frac{x(1-x)}{d_r(x)},\qquad a_r=\frac{r-1}{r}.
\]
Since $a_r'=1/r^2$ and $d_r(a_r)=1/[r^2(r+1)]$, the endpoint term in Leibniz's rule is
\[
f_r(a_r)a_r'=1-\frac1{r^2}.
\]
The integrand increases with $r$:
\[
\partial_r f_r(x)
=\frac{x(1-x)}{(r+1)^2d_r(x)^2}>0
\qquad(0<x<a_r).
\]
Therefore
\begin{equation}\label{eq:J-prime-exact}
J'(r)=1-\frac1{r^2}
+\frac1{(r+1)^2}\int_0^{a_r}
\frac{x(1-x)}{d_r(x)^2}\dd x
>1-\frac1{r^2}.
\end{equation}
Also
\[
\frac{\dd}{\dd r}\frac{r-1}{r(r+1)}
=\frac{-r^2+2r+1}{r^2(r+1)^2}.
\]
Consequently,
\[
H_+'(r)>
1-\frac1{r^2}+\frac{-r^2+2r+1}{r^2(r+1)^2}
=\frac{r^2+2r-1}{(r+1)^2}>0.
\]
As $r\downarrow1$, both terms in \cref{eq:Hplus} tend to zero. As $r\uparrow3$, $J(r)\to\infty$ because $d_r$ develops a double zero at $x=1/2$ while $x(1-x)$ stays positive there. Thus the two monotone functions cross exactly once. Substitution of \cref{eq:Hminus-Theta,eq:Hplus-Theta} gives \cref{eq:root-equation}.
\end{proof}

Set
\begin{equation}\label{eq:hstar-alpha}
h_\star:=H_-(r_\star)=H_+(r_\star),\qquad
\alpha_\star:=\frac1{1+h_\star}.
\end{equation}
By \cref{thm:V-exact}, $V(h_\star)=h_\star$. Also \cref{eq:Hplus-Theta} gives
\begin{equation}\label{eq:alpha-formula}
\alpha_\star=\frac{r_\star+1}{\Theta(r_\star)+2}.
\end{equation}

\subsection{The universal guarantee}

\begin{theorem}[Universal fixed-price guarantee]\label{thm:universal}
For every pair of independent, nonnegative finite-mean distributions, $\FP(F_S,F_B)\ge\alpha_\star\OPT(F_S,F_B)$.
\end{theorem}

\begin{proof}
If $M=0$, then $I=0$ and the ratio is one. Suppose $M>0$ and divide all values by $M$. By \cref{cor:universal-reduction} and $V(h_\star)=h_\star$, $I-h_\star m\le1+h_\star$.
Undoing the scaling gives $I-h_\star m\le(1+h_\star)M$.
Therefore $m+I\le(1+h_\star)(m+M)$. Using \cref{lem:welfare-decomp} proves the claim.
\end{proof}

\subsection{Numerical value}\label{sec:numerics}

The characterization in \cref{eq:root-equation,eq:alpha-formula} is exact, and the integral $\Theta$ has an elementary closed form. Let
\begin{equation}\label{eq:gamma}
\gamma(r):=\sqrt{\frac{3-r}{r+1}}.
\end{equation}
Completing the square in \cref{eq:theta-intro} gives
\begin{equation}\label{eq:Theta-closed}
\Theta(r)=\frac2{\gamma(r)}
\left[
\arctan\left(\frac{r-2}{r\gamma(r)}\right)
+\arctan\left(\frac1{\gamma(r)}\right)
\right].
\end{equation}
This is the form used for numerical evaluation.

Interval arithmetic applied to \cref{eq:Theta-closed,eq:root-equation} gives
\begin{align}
1.53963634239276&<r_\star<1.53963634239278,\notag\\
0.73802433573449&<\afp<0.73802433573450.\label{eq:alpha-interval}
\end{align}
Thus $\afp=0.73802433573449\ldots$,
where the displayed digits are exactly those determined by the enclosure \cref{eq:alpha-interval}.
The interval bounds above are only a numerical evaluation of the analytic formulas. The theorem itself is the exact root characterization in \cref{eq:root-equation,eq:alpha-formula}, together with the uniqueness proof in \cref{lem:unique-fixed-point}.

\section{Matching extremal distributions}\label{sec:extremal}

We now show that the constant of \cref{thm:universal} cannot be improved. The plan is to read the optimal control backwards: as the survival function of a bounded buyer body, together with the saturated seller of \cref{sec:saturation}. This pair realizes the control value exactly, but the buyer's call transform lies exactly one unit below the call coordinate of the control (\cref{eq:B0-call} below). A vanishing buyer atom at a diverging value restores the missing unit of first moment; \cref{fig:extremal} shows both components.

We first specialize to the fixed point $h=h_\star$. All formulas below are parameter-uniform in $h>0$ and $r=r(h)$; this observation will be used in \cref{prop:decision-equivalence}.

Write
\begin{equation}\label{eq:extremal-parameters}
r:=r_\star,\quad h:=h_\star,\quad a:=\frac{r-1}{r},\quad
\kappa:=\frac1{r+1},\quad d(x):=x^2-x+\kappa.
\end{equation}
Define, for $0\le x\le a$,
\begin{equation}\label{eq:c-x-extremal}
c(x):=\exp\left(\int_0^x\frac{\xi}{d(\xi)}\dd\xi\right).
\end{equation}
Set
\begin{equation}\label{eq:s-b-extremal}
s:=c(a)=e^{T(r)},\qquad b:=rs.
\end{equation}
On $1\le c\le s$, define $q_\star$ parametrically by
\begin{equation}\label{eq:qstar-parametric}
q_\star(c(x)):=hc(x)\frac{x}{d(x)},\qquad0\le x\le a,
\end{equation}
and set $q_\star(c):=1$ for $s\le c\le b$. The identity $hs=1/[r(r^2-1)]$ implies $q_\star(s)=1$.

\begin{lemma}[Monotonicity of the extremal control]\label{lem:qstar-monotone}
The function $q_\star$ is continuous and nondecreasing on $[1,b]$, strictly increasing on $[1,s]$, with $q_\star(1)=0$ and $q_\star(s)=q_\star(b)=1$.
\end{lemma}

\begin{proof}
For $0<x<a$,
\begin{equation}\label{eq:qstar-monotone-derivative}
\frac{\dd}{\dd x}\log q_\star(c(x))
=\frac{x}{d(x)}+\frac1x-\frac{2x-1}{d(x)}
=\frac1x+\frac{1-x}{d(x)}>0.
\end{equation}
The endpoint values follow from \cref{eq:qstar-parametric} and the switch relation.
\end{proof}

Together, \cref{eq:c-x-extremal,eq:qstar-parametric,eq:s-b-extremal} reconstruct the unique one-switch trajectory of \cref{thm:V-exact} with parameter $r=r_\star$. Hence $(b,q_\star)$ is the unique free-endpoint maximizer for $V(h)$, and
\begin{equation}\label{eq:qstar-optimal-value}
\K_h(b,q_\star)=V(h)=h.
\end{equation}

Define
\begin{equation}\label{eq:p-c-extremal}
p(c):=\int_c^b\frac{\dd t}{q_\star(t)},\qquad1\le c\le b,
\end{equation}
and let $L:=p(1)<\infty$. Finiteness follows from \cref{lem:endpoint-regularity}; directly, on the interior arc $\dd c/q_\star(c)=\dd x/h$. Let $c_\star\colon[0,L]\to[1,b]$ be the inverse of the decreasing map $c\mapsto p(c)$. Thus $c_\star(p)$ is the call coordinate attached to the price $p$; it is not the call transform of the buyer constructed below, which is $c_\star-1$ by \cref{eq:B0-call}. Closed-form expressions for the resulting buyer body and seller, parametrized by $x$, are given in \cref{sec:explicit-extremal}.

\subsection{The bounded buyer component}

Define the right-continuous survival function $Q_{B_0}(p)=\Pp(B_0>p)$ by
\begin{equation}\label{eq:B0-survival}
Q_{B_0}(p):=
\begin{cases}
q_\star(c_\star(p)),&0\le p<L,\\
0,&p\ge L.
\end{cases}
\end{equation}
By \cref{lem:qstar-monotone}, $Q_{B_0}$ is continuous and nonincreasing, with $Q_{B_0}(0)=1$ and $Q_{B_0}(L-)=Q_{B_0}(L)=0$. It is therefore the survival function of a random variable $B_0$ supported on $[0,L]$; in particular, $B_0$ has no endpoint atoms at $0$ or $L$. Its call transform is
\begin{equation}\label{eq:B0-call}
C_{B_0}(p):=\E[(B_0-p)_+]=c_\star(p)-1,\qquad0\le p\le L,
\end{equation}
because both sides are absolutely continuous and vanish at $L$, and almost everywhere both have the derivative $-Q_{B_0}(p)=-q_\star(c_\star(p))$.

\subsection{The bounded seller}

Define
\begin{equation}\label{eq:Rstar}
R_\star(c):=\int_c^b\frac{\dd t}{q_\star(t)t^2}.
\end{equation}
For $0\le p\le L$, let
\begin{equation}\label{eq:Astar}
A_\star(p):=c_\star(p)R_\star(c_\star(p)),
\end{equation}
and extend with slope one for $p\ge L$. The following calculation shows that this is a seller put transform. Since $q_\star$ is nondecreasing in $c$ while $c_\star(p)$ is decreasing in $p$, the derivative $c_\star'(p)=-q_\star(c_\star(p))$
is nondecreasing; hence $c_\star$ is convex. At every differentiability point of $c_\star$, differentiating \cref{eq:Astar} gives
\begin{equation}\label{eq:Fstar-formula}
A_\star'(p)=\frac{1}{c_\star(p)}-q_\star(c_\star(p))R_\star(c_\star(p)).
\end{equation}
Let $F_\star:=A'_{\star,+}$ be the right-continuous version of this derivative. Writing $F_\star=c_\star'\,R_\star(c_\star)+1/c_\star$, the distributional product rule gives
\begin{align}
\dd F_\star
&=R_\star(c_\star)\,\dd c_\star'+c_\star'\,\dd(R_\star(c_\star))+\dd(1/c_\star)\notag\\
&=R_\star(c_\star)\,\dd c_\star'+\frac{c_\star'}{c_\star^2}\dd p-\frac{c_\star'}{c_\star^2}\dd p
=R_\star(c_\star)\,\dd c_\star'\ge0.
\label{eq:Fstar-monotone}
\end{align}
Thus $F_\star$ is nondecreasing. At $p=0$, $R_\star(b)=0$, so $F_\star(0)=1/b$; as $p\uparrow L$, we have $c_\star(p)\downarrow1$, $q_\star(c_\star(p))\downarrow0$, and hence $F_\star(p)\uparrow1$. Therefore the slope-one extension is convex, its right derivative is a right-continuous CDF on $[0,L]$, and \cref{lem:put-characterization} shows that it is the put transform of a random variable $S_\star$ supported on $[0,L]$. The seller has an atom of mass $1/b$ at zero and no atom at $L$. Write $m_\star:=\E[S_\star]=L-A_\star(L)$.

With the conventions just fixed, $F_\star(p)=\Pp(S_\star\le p)$ and $Q_{B_0}(p)=\Pp(B_0>p)$. The pointwise saturation identity is $F_\star(p)c_\star(p)+A_\star(p)Q_{B_0}(p)=1$ for $0\le p\le L$.
Since $C_{B_0}=c_\star-1$, the gain captured from the bounded buyer component at price $p$ is
\begin{equation}\label{eq:g0}
\E[(B_0-S_\star)\one\{S_\star\le p\le B_0\}]=1-F_\star(p).
\end{equation}

Let $G_\star:=\E[(B_0-S_\star)_+]$ be the first-best gain of the bounded pair $(S_\star,B_0)$. It is the quantity $I$ of \cref{eq:m-I-g-M} for this auxiliary pair. We give it its own letter because the instances of interest are the pairs $(S_\star,B_\varepsilon)$ below, whose first-best gain tends to $1+G_\star$ rather than to $G_\star$. By independence and \cref{eq:I-identities},
\begin{align}
G_\star
&=\int_0^L F_\star(p)Q_{B_0}(p)\dd p
 =\int_1^b F_\star(p(c))\dd c \notag\\
&=\log b-\int_1^b q_\star(c)
 \left(\int_c^b\frac{\dd t}{q_\star(t)t^2}\right)\dd c \notag\\
&=\log b-\int_1^b
 \frac{Y(c)}{q_\star(c)c^2}\dd c,
\label{eq:Gstar-explicit}
\end{align}
where $Y(c)=\int_1^c q_\star(\sigma)\dd\sigma$ and Tonelli's theorem is used in the last step. Similarly, the slope-one extension gives
\begin{equation}\label{eq:mstar-explicit}
m_\star=\int_1^b\frac{1-c^{-2}}{q_\star(c)}\dd c.
\end{equation}
Subtracting $h$ times \cref{eq:mstar-explicit} from \cref{eq:Gstar-explicit} and using \cref{eq:K-definition} yields
\begin{equation}\label{eq:Gstar-control}
G_\star-hm_\star=\K_h(b,q_\star).
\end{equation}
Using \cref{eq:qstar-optimal-value}, we obtain
\begin{equation}\label{eq:Gstar-key}
G_\star=h(1+m_\star).
\end{equation}

\subsection{Explicit formulas for the extremal pair}\label{sec:explicit-extremal}

The bounded components can be written without any inverse function, using the parameter $x$ of \cref{eq:c-x-extremal}.
On the interior arc, \cref{eq:c-x-extremal,eq:qstar-parametric} imply
\begin{equation}\label{eq:dp-dx}
\frac{\dd c}{q_\star(c)}=\frac{\dd x}{h}.
\end{equation}
Let $p_s:=b-s=s(r-1)$ denote the length of the terminal $q=1$ arc. Then for $0\le x\le a$,
\begin{equation}\label{eq:p-x-explicit}
p(c(x))=p_s+\frac{a-x}{h}.
\end{equation}
In particular,
\begin{equation}\label{eq:L-explicit}
L=s(r-1)+\frac ah.
\end{equation}
The bounded buyer survival is
\begin{equation}\label{eq:Q0-explicit}
Q_{B_0}(p)=
\begin{cases}
1,&0\le p\le p_s,\\[1mm]
hc(x(p))\dfrac{x(p)}{d(x(p))},&p_s<p<L,\\[2mm]
0,&p\ge L,
\end{cases}
\end{equation}
where $x(p):=a-h(p-p_s)$, so that $c_\star(p)=c(x(p))$ for $p_s<p<L$. The seller CDF is
\begin{equation}\label{eq:Fstar-explicit}
F_\star(p)=\frac1{c_\star(p)}-Q_{B_0}(p)\int_{c_\star(p)}^b\frac{\dd t}{q_\star(t)t^2},
\qquad0\le p<L,
\end{equation}
with $F_\star(p)=1$ for $p\ge L$. These formulas define the distributions without any limiting or optimization operation; only the final buyer tail mixture varies with $\varepsilon$. For scale, numerical evaluation gives $b\approx2.0555924272089955$ and $L\approx1.7078760681197155$.

\subsection{The escaping buyer tail}

For $0<\varepsilon<1/L$, define
\begin{equation}\label{eq:Bepsilon}
B_\varepsilon:=
\begin{cases}
B_0,&\text{with probability }1-\varepsilon,\\
1/\varepsilon,&\text{with probability }\varepsilon.
\end{cases}
\end{equation}
Take $B_\varepsilon$ independent of $S_\star$. Every $B_\varepsilon$ is nonnegative and has finite first moment.
As $\varepsilon\to0$, $B_\varepsilon$ converges to $B_0$ in distribution. Its first moment does not converge to $\E[B_0]$, because the remote atom contributes $\varepsilon\cdot(1/\varepsilon)=1$ for every $\varepsilon$ (\cref{fig:extremal}(b)). This escaping unit supplies the missing unit of call value: for $0\le p\le L$, the call transform of $B_\varepsilon$ is $(1-\varepsilon)C_{B_0}(p)+1-\varepsilon p$, which tends to $c_\star(p)$ as $\varepsilon\to0$.

\begin{figure}[!htb]
\centering
% Figure: the matching extremal family.
% (a) exact distribution functions at h = h_* (data from gen_figure_data.py):
%     p_s = b - s = 0.720477, L = 1.707876, 1/b = 0.486478.
% (b) schematic valuation axis for B_eps (axis break; heights are not to scale).
\begin{tikzpicture}
  % ---------- panel (a): bounded components ----------
  \begin{scope}[x=3.6cm, y=3.3cm]
    \def\ps{0.720477}
    \def\L{1.707876}
    \def\ib{0.486478}
    \draw[fp guide] (0,1) -- (\L,1);
    \draw[fp guide] (\ps,0) -- (\ps,1);
    \draw[fp guide] (\L,0) -- (\L,1);
    \draw[fp main] plot file {figures/data/buyer_survival.dat};
    \draw[fp accent dashed] plot file {figures/data/seller_cdf.dat};
    \draw[fp axis] (0,0) -- (1.92,0) node[right, fp label] {$p$};
    \draw[fp axis] (0,0) -- (0,1.22);
    \draw[fp tick] (0,1) -- (-0.03,1) node[left, fp label] {$1$};
    \draw[fp tick] (0,\ib) -- (-0.03,\ib) node[left, fp label] {$1/b$};
    \node[below, fp label] at (0,0) {$0$};
    \draw[fp tick] (\ps,0) -- (\ps,-0.035) node[below, fp label] {$p_s$};
    \draw[fp tick] (\L,0) -- (\L,-0.035) node[below, fp label] {$L$};
    % seller atom at 0: CDF jumps from 0 to 1/b
    \node[fp open, draw=fpblue] at (0,0) {};
    \node[fp atom, fill=fpblue] at (0,\ib) {};
    \node[fp note, anchor=south west] at (0.07,1.02) {$\Pr(B_0>p)$};
    \node[fp note, text=fpblue, anchor=south west] at (0.07,\ib+0.02) {$\Pr(S_\star\le p)$};
    \node[fp note, text=fpblue, anchor=west] at (0.16,0.18) {atom $1/b$ at $0$};
    \draw[fpblue, line width=0.4pt, -{Stealth[length=1.6mm]}] (0.15,0.18) -- (0.035,\ib-0.06);
    \node[fp panel] at (0.95,-0.52) {(a) bounded seller and buyer body};
  \end{scope}
  % ---------- panel (b): the buyer B_eps ----------
  \begin{scope}[shift={(7.9cm,0)}, x=1cm, y=1cm]
    \def\Lb{2.2}   % position of L (schematic)
    \def\At{5.5}   % position of 1/eps (schematic)
    % bounded body: a support band (NOT a density), total mass 1 - eps
    \fill[fp region] (0,0) rectangle (\Lb,0.22);
    \node[fp note] at ({\Lb/2},0.55) {$B_0$ with probability $1-\varepsilon$};
    \draw[decorate, decoration={brace, amplitude=4pt, mirror}] (0,-0.42) -- (\Lb,-0.42)
      node[midway, below=5pt, fp label] {support $[0,L]$};
    % axis with a break
    \draw[line width=0.5pt] (0,0) -- (3.35,0);
    \draw[line width=0.5pt] (3.55,0) -- (\At-0.0,0);
    \draw[fp axis] (\At,0) -- (6.4,0) node[right, fp label] {value};
    \draw[line width=0.5pt] (3.25,-0.13) -- (3.45,0.13);
    \draw[line width=0.5pt] (3.45,-0.13) -- (3.65,0.13);
    \draw[fp tick] (0,0) -- (0,-0.1) node[below, fp label] {$0$};
    \draw[fp tick] (\Lb,0) -- (\Lb,-0.1);
    % remote atom: mass eps at value 1/eps
    \draw[fpred, line width=1.1pt] (\At,0) -- (\At,1.5);
    \node[fp atom, fill=fpred] at (\At,1.5) {};
    \node[fp label, text=fpred, right=2pt] at (\At,1.5) {mass $\varepsilon$};
    \draw[fp tick] (\At,0) -- (\At,-0.1) node[below, fp label] {$1/\varepsilon$};
    \node[fp note, anchor=south] at (\At,1.75) {$\varepsilon\cdot\frac1\varepsilon=1$};
    \node[fp note, anchor=north] at (\At,-0.62) {as $\varepsilon\to0$: mass $\to0$,\\first moment $=1$};
    \node[fp panel] at (3.2,-1.716) {(b) buyer $B_\varepsilon$ (schematic)};
  \end{scope}
\end{tikzpicture}
\caption{The matching extremal family. (a)~Exact distribution functions of the bounded components at $h=h_\star$: the buyer survival function $\Pr(B_0>p)$ (solid; equal to $q_\star(c_\star(p))$) and the seller distribution function $\Pr(S_\star\le p)$ (blue, dashed). The seller has an atom of mass $1/b$ at $0$ and no further mass on $(0,p_s]$, where $p_s=b-s$ is the length of the boundary arc. Both components are supported on $[0,L]$. (b)~Schematic of the buyer $B_\varepsilon$, with a broken axis and heights not to scale. It takes the body $B_0$ with probability $1-\varepsilon$ and the value $1/\varepsilon$ with probability $\varepsilon$. The remote atom vanishes in distribution but contributes exactly one unit of first moment for every $\varepsilon$. This escaping unit is what makes the ratio converge to $\afp$, and its necessity is why no single instance attains $\afp$ (\cref{prop:distribution-nonattainment}).}\label{fig:extremal}
\end{figure}

Let $g_\varepsilon$, $M_\varepsilon$, and $I_\varepsilon$ denote the quantities $g$, $M$, and $I$ of \cref{eq:m-I-g-M} for the pair $(S_\star,B_\varepsilon)$. For $0\le p\le L$, define $\mu_\star(p):=\E[S_\star\one\{S_\star\le p\}]$. By \cref{eq:g0},
\begin{align}
g_\varepsilon(p)
&=(1-\varepsilon)(1-F_\star(p))
+\varepsilon\E[(1/\varepsilon-S_\star)\one\{S_\star\le p\}]\notag\\
&=1-\varepsilon\bigl(1-F_\star(p)+\mu_\star(p)\bigr)\le1.
\label{eq:g-epsilon-low-price}
\end{align}
For every $L\le p\le1/\varepsilon$, the bounded buyer cannot trade and all sellers lie below $p$, so $g_\varepsilon(p)=1-\varepsilon m_\star$.
For $p>1/\varepsilon$, the gain is zero. Therefore
\begin{equation}\label{eq:M-epsilon-sandwich}
1-\varepsilon m_\star\le M_\varepsilon=\max_{p\ge0}g_\varepsilon(p)\le1,
\end{equation}
and $M_\varepsilon\to1$.

The first-best gain of $(S_\star,B_\varepsilon)$ is
\begin{equation}\label{eq:I-epsilon}
I_\varepsilon=(1-\varepsilon)G_\star+1-\varepsilon m_\star,
\end{equation}
because $S_\star\le L<1/\varepsilon$. Thus $I_\varepsilon\to1+G_\star$. Using \cref{lem:welfare-decomp,eq:Gstar-key},
\begin{align}
\lim_{\varepsilon\downarrow0}
\frac{\FP(F_{S_\star},F_{B_\varepsilon})}{\OPT(F_{S_\star},F_{B_\varepsilon})}
&=\frac{m_\star+1}{m_\star+1+G_\star}
=\frac1{1+h}=\alpha_\star.
\label{eq:extremal-limit}
\end{align}
This proves the extremal side of \cref{thm:main}.

\begin{remark}[The remote buyer atom]\label{rem:remote-atom}
The idea of placing a very large buyer value with very small probability already appears in the analysis of Liu, Ren, and Wang~\cite{LiuRenWang2023}, where it realizes a boundary case of their formulation. Remote buyer support points also appear in the finite hard instances of Giambartolomei and de Keijzer~\cite{GiambartolomeiDeKeijzer2026}. In the present analysis the remote atom is a forced feature of the exact worst case: \cref{prop:distribution-nonattainment} shows that no instance without escaping first moment can attain the bound. Whether the finite instances of~\cite{GiambartolomeiDeKeijzer2026} converge to the continuous optimizer is a separate question that we do not address.
\end{remark}

\subsection{Non-attainment and proof of the main theorem}

\begin{proposition}[Non-attainment of the worst-case infimum]\label{prop:distribution-nonattainment}
No admissible pair of finite-mean distributions attains $\afp$.
\end{proposition}

\begin{proof}
Suppose an admissible pair attained the ratio $\alpha_\star=(1+h_\star)^{-1}$. It cannot have $M=0$, because then $I=0$ and its ratio is one. Normalize to $M=1$. Equality of the welfare ratio is equivalent to
\begin{equation}\label{eq:nonattainment-equality}
I-h_\star m=1+h_\star.
\end{equation}
The case $C(0)\le1$ is impossible, since then $I-h_\star m\le1<1+h_\star$. Thus $C(0)>1$, and the universal proof gives the chain
\begin{equation}\label{eq:nonattainment-chain}
I-h_\star m
\le \widetilde I-h_\star\widetilde m
=1+\K_{h_\star}(C(0),q)
\le1+V(h_\star)=1+h_\star.
\end{equation}
By \cref{eq:nonattainment-equality}, every inequality in \cref{eq:nonattainment-chain} is an equality. Hence the buyer-generated control $q$ is a free-endpoint maximizer of $V(h_\star)$. The uniqueness in \cref{thm:V-exact,lem:trajectory-reconstruction} forces $q=q_\star$ almost everywhere, and therefore $q(c)\to0$ as $c\downarrow1$.

On the other hand, let $p_1$ be defined by $C(p_1)=1$. Since $1=C(p_1)=\E[(B-p_1)_+]$,
we have $\Pp(B>p_1)>0$. For $p<p_1$, $Q_B(p)\ge Q_B(p_1)>0$, so the inverse-call control $q(c)=Q_B(p(c))$ is bounded away from zero as $c\downarrow1$. This contradicts equality almost everywhere with the continuous control $q_\star$, which is arbitrarily small on an interval of positive length near $c=1$.
\end{proof}

\begin{proof}[Proof of \cref{thm:main}]
The uniqueness of $r_\star$ and the exact formula for $\alpha_\star$ follow from \cref{lem:unique-fixed-point,eq:alpha-formula}. \Cref{thm:universal} gives part~(a) and $\afp\ge\alpha_\star$, while the limiting family in \cref{eq:extremal-limit} gives part~(b) and $\afp\le\alpha_\star$. Thus $\afp=\alpha_\star$. Part~(c) combines \cref{prop:distribution-nonattainment} with \cref{lem:price-attainment}.
\end{proof}

\subsection{Exactness of the reduction at every threshold}

The construction of this section did not use the fixed-point property of $h_\star$ except in the final evaluation \cref{eq:extremal-limit}. Running it at an arbitrary $h>0$ shows that the control relaxation of \cref{sec:control-reduction} loses nothing at any approximation threshold.

\begin{proposition}[Threshold equivalence]\label{prop:decision-equivalence}
For every $h>0$,
\begin{equation}\label{eq:decision-equivalence}
\afp\ge\frac1{1+h}
\quad\Longleftrightarrow\quad
V(h)\le h,
\end{equation}
where $V(h)$ is defined in \cref{eq:V-definition}.
\end{proposition}

Sufficiency of $V(h)\le h$ is the universal reduction, \cref{cor:universal-reduction}. Necessity follows by realizing the optimizing control at parameter $h$ and appending the escaping atom. The decision version of the problem studied by Liu, Ren, and Wang~\cite[Theorem~3.3]{LiuRenWang2023} has the same logical form: a candidate ratio is valid if and only if a one-dimensional functional of the buyer distribution is bounded. With $\beta=(1+h)^{-1}$, \cref{prop:decision-equivalence} is the corresponding statement in our coordinates.

\begin{proof}[Proof of \cref{prop:decision-equivalence}]
\emph{Sufficiency.} If $V(h)\le h$, normalize an arbitrary instance with $M>0$ to $M=1$. By \cref{cor:universal-reduction}, $I-hm\le1+V(h)\le1+h$.
Undoing the scaling and applying \cref{lem:linear-target} gives the left-hand side of \cref{eq:decision-equivalence}. The case $M=0$ has ratio one.

\emph{Necessity.} Let $r=r(h)$ and take the unique optimizing control from \cref{lem:trajectory-reconstruction}. The monotonicity calculation in \cref{eq:q-reconstruction-monotone} realizes it as a bounded buyer body and a bounded seller exactly as in \cref{eq:p-c-extremal,eq:B0-survival,eq:Astar}. Denote their seller mean and first-best gain by $m_h$ and $G_h$. The calculation in \cref{eq:Gstar-explicit,eq:mstar-explicit,eq:Gstar-control} is parameter-uniform and gives $G_h-hm_h=V(h)$.
Append probability $\varepsilon$ at buyer value $1/\varepsilon$. The same bounds as in \cref{eq:M-epsilon-sandwich,eq:I-epsilon} give $M_{h,\varepsilon}\to1$ and $I_{h,\varepsilon}\to1+G_h$. Hence, as $\varepsilon\downarrow0$, the welfare ratios converge to
\begin{equation}\label{eq:generic-ratio}
\frac{1+m_h}{1+m_h+G_h}
=\frac{1+m_h}{1+V(h)+(1+h)m_h},
\end{equation}
and a direct subtraction gives
\begin{equation}\label{eq:generic-sign}
\frac{1+m_h}{1+V(h)+(1+h)m_h}-\frac1{1+h}
=\frac{h-V(h)}{(1+h)(1+V(h)+(1+h)m_h)}.
\end{equation}
If $V(h)>h$, the limit \cref{eq:generic-ratio} is smaller than $1/(1+h)$, so admissible members of the limiting family have ratio below $1/(1+h)$ for all sufficiently small $\varepsilon$. The left-hand side of \cref{eq:decision-equivalence} therefore fails.
\end{proof}

\section{Conclusion}

We determined the exact worst-case welfare ratio of the optimal fixed-price mechanism in bilateral trade with independent, fully known value distributions. It equals $(r_\star+1)/(\Theta(r_\star)+2)=0.7380243357\ldots$, which closes the interval left by previous work. The universal guarantee holds for arbitrary finite-mean distributions, and an explicit family of instances shows that it is tight.

The proof rests on two observations. After the seller is eliminated by saturating the Wronskian constraint, the problem becomes a one-dimensional control problem that is strictly convex in the logarithm of the control, so a first-order analysis yields a certified global optimum. The optimizer switches once from an interior arc to the boundary and can be integrated in closed form. The extremal structure also separates two attainment phenomena: every fixed instance has an optimal price, while the worst-case ratio is approached only through escape of first moment to infinity.

\newcommand{\etalchar}[1]{$^{#1}$}

\end{document}